\documentclass[12pt]{amsart}
\usepackage{a4wide}
\usepackage[latin9]{inputenc}
\usepackage{amsmath}
\usepackage{amsfonts}
\usepackage{amssymb}
\usepackage{amsthm}
\usepackage{subfigure}
\usepackage[usenames,dvipsnames]{color}
\usepackage{pstricks}
\usepackage{graphicx}
\usepackage[all]{xy}
\newtheorem{theo}{Theorem}[section]
\newtheorem{prop}[theo]{Proposition}
\newtheorem{coro}[theo]{Corollary}
\newtheorem{lemm}[theo]{Lemma}

\theoremstyle{definition}
\newtheorem{def1}[theo]{Definition}
\theoremstyle{remark}
\newtheorem{rema}[theo]{Remark}
\newcommand{\Op}{\operatorname{Op}}

\newcommand{\nwc}{\newcommand}
\nwc{\eps}{\epsilon}
\nwc{\ep}{\epsilon}
\nwc{\vareps}{\varepsilon}
\nwc{\Oph}{\operatorname{Op}_\hbar}
\nwc{\la}{\langle}
\nwc{\ra}{\rangle}

\nwc{\mf}{\mathbf} 
\nwc{\blds}{\boldsymbol} 
\nwc{\ml}{\mathcal} 

\nwc{\defeq}{\stackrel{\rm{def}}{=}}

\nwc{\cE}{\ml{E}}
\nwc{\cN}{\ml{N}}
\nwc{\cO}{\ml{O}}
\nwc{\cP}{\ml{P}}
\nwc{\cU}{\ml{U}}
\nwc{\cV}{\ml{V}}
\nwc{\cW}{\ml{W}}
\nwc{\tU}{\widetilde{U}}
\nwc{\IN}{\mathbb{N}}
\nwc{\IR}{\mathbb{R}}
\nwc{\IZ}{\mathbb{Z}}
\nwc{\IC}{\mathbb{C}}
\nwc{\IT}{\mathbb{T}}
\nwc{\IS}{\mathbb{S}}
\nwc{\tP}{\widetilde{P}}
\nwc{\tPi}{\widetilde{\Pi}}
\nwc{\tV}{\widetilde{V}}
\nwc{\supp}{\operatorname{supp}}
\nwc{\rest}{\restriction}

\begin{document}

\title[Topology of Pollicott-Ruelle resonant states]{Topology of Pollicott-Ruelle resonant states}

\author[Nguyen Viet Dang]{Nguyen Viet Dang}

\address{Institut Camille Jordan (U.M.R. CNRS 5208), Universit\'e Claude Bernard Lyon 1, B\^atiment Braconnier, 43, boulevard du 11 novembre 1918, 
69622 Villeurbanne Cedex }

\email{dang@math.univ-lyon1.fr}

\author[Gabriel Rivi\`ere]{Gabriel Rivi\`ere}

\address{Laboratoire Paul Painlev\'e (U.M.R. CNRS 8524), U.F.R. de Math\'ematiques, Universit\'e Lille 1, 59655 Villeneuve d'Ascq Cedex, France}

\email{gabriel.riviere@math.univ-lille1.fr}

\begin{abstract} 
We prove that the twisted De Rham cohomology of a flat vector bundle over some smooth manifold 
is isomorphic to the cohomology of invariant Pollicott--Ruelle resonant states associated with Anosov and 
Morse--Smale flows. As a consequence, we obtain generalized Morse inequalities for such flows. In the case of Morse--Smale flows, we relate 
the resonances lying on the imaginary axis with the twisted Fuller measures used by Fried in his work on Reidemeister torsion. In particular, when $V$ is a nonsingular 
Morse-Smale flow, we 
show that the Reidemeister torsion can be recovered from
the only knowledge of dynamical
resonances on the imaginary axis by expressing the torsion  
as a zeta regularized infinite product of these  
resonances.
\end{abstract}

\maketitle

\section{Introduction}

Consider $M$ a smooth ($\ml{C}^{\infty}$), compact and oriented manifold which has no boundary and which is of dimension $n\geq 1$. We say that 
$f\in\ml{C}^{\infty}(M,\IR)$ is a Morse function if it has finitely many critical points all of them being nondegenerate. Denote by $c_k(f)$ the 
number of critical points of index $k$ and by $b_k(M)$ the $k$-th Betti number of the manifold. A famous result of Morse states that the following 
inequalities hold~\cite{Mo25}:
$$\forall 0\leq k\leq n,\quad \sum_{j=0}^k(-1)^{k-j}b_j(M)\leq\sum_{j=0}^k(-1)^{k-j}c_j(f),$$
with equality when $k=n$. These inequalities are usually called Morse inequalities. Later on, Thom~\cite{Th49} and Smale~\cite{Sm60} 
observed that, if we are given a Riemannian metric on $M$, then we can define a cell decomposition of the manifold by considering the stable 
manifolds associated with each critical point by the induced gradient flow. Then, $c_j(f)$ is thought as the number of stable manifolds of 
dimension $j$ and Smale gave a new proof of these inequalities~\cite{Sm60} which had in some sense a more dynamical flavour. This dynamical 
approach to Morse inequalities was then pursued by many people and generalized to much more general dynamical systems. Among others, we refer 
the reader to the books of Conley~\cite{Co78} and Franks~\cite{Fr82} for examples of such generalizations. This dynamical approach to differential 
topology also leads to dynamical interpretation of other topological invariants such as the Reidemeister torsion by Fried~\cite{F87}. 
Inspired 
by the theory of currents, but still from a dynamical perspective, we can also mention the works of Laudenbach~\cite{Lau92} and 
Harvey--Lawson~\cite{HaLa01} who showed how to realize the cell decomposition of Thom and Smale in terms of De Rham currents.

More than fifty years after Morse's seminal work, Witten introduced in the context of Hodge theory another approach to Morse 
inequalities~\cite{Wi82}. This point of view was further developped by Helffer and Sj\"ostrand using tools from semiclassical 
analysis~\cite{HeSj85} and brought a new spectral perspective on these results. In particular, the coefficients $c_j(f)$ are interpreted by 
Witten as the number of small eigenvalues (counted with multiplicity) of a certain deformation of the Hodge Laplacian acting on forms of 
degree $j$. This also lead to many developments in analysis and topology that would be again hard to describe in details. We can for 
instance quote the works of Bismut-Zhang~\cite{BisZhang92} and Burghelea-Friedlander-Kappeler~\cite{burg96} who developped the 
relationship
between Witten's approach and Reidemeister torsion.

If we come back to the dynamical perspective, a spectral approach slightly different from the ones arising from Hodge-Witten theory 
was recently developped by several authors to study dynamical systems with hyperbolic behaviour. We shall describe these results more 
precisely below. In particular, this theory applies to the dynamical systems 
studied by Smale~\cite{Sm67}: Anosov flows~\cite{BuLi07, FaSj11, DyZw13, FaTs17}, Axiom A flows~\cite{DyGu14}
and
Morse-Smale flows~\cite{DaRi16, DaRi17a}. 
These different works were initially motivated by the study of the correlation function in dynamical systems 
and by the analytic properties of the so-called Ruelle zeta function~\cite{Ru76}. In the present work, 
we aim at showing how these results can be also used to obtain classical results from differential topology in a manner which is somewhat 
intermediate between the spectral approach of Witten and the more dynamical one of Thom and Smale. As a by-product, it will also unveil the topological informations contained 
in the correlation spectrum (also called Pollicott-Ruelle spectrum),
resonances on the imaginary axis in particular,
studied in the above works.

\section{Statement of the main results}
\label{s:mainresults}
\subsection{Dynamical framework}
Let $M$ be a smooth, compact, oriented, boundaryless manifold of dimension $n\geq 1$ and let $\mathbf{p}:\ml{E}\rightarrow M$ be a smooth complex vector 
bundle of rank $N$. Suppose that $\ml{E}$ is endowed with a flat connection $\nabla:\Omega^0(M,\ml{E})\rightarrow\Omega^1(M,\ml{E})$~\cite[Ch.~12]{Lee09}. 
This allows us to define an exterior derivative $d^{\nabla}:\Omega^{\bullet}(M,\ml{E})\rightarrow\Omega^{\bullet+1}(M,\ml{E})$ which satisfies 
$d^{\nabla}\circ d^{\nabla}=0$. In particular, one can introduce the twisted De Rham complex $(\Omega^*(M,\ml{E}),d^{\nabla})$ associated with $d^{\nabla}$:
$$0\xrightarrow{d^{\nabla}} \Omega^0(M,\ml{E})\xrightarrow{d^{\nabla}} \Omega^{1}(M,\ml{E})\xrightarrow{d^{\nabla}} \ldots\xrightarrow{d^{\nabla}} 
\Omega^n(M,\ml{E})\xrightarrow{d^{\nabla}} 0.$$
The $k$th-cohomology group of this complex is denoted by $\mathbf{H}^k(M,\ml{E})$. This kind of 
complex appears for instance naturally in Hodge theory~\cite{BDIP96, Mn14}.

Fix now $V$ a smooth vector field on $M$ and denote by $\varphi^t$ the induced flow on $M$. Using the flat connection $\nabla$, one can lift this flow 
into a flow $\Phi_k^t$ on the vector bundle $\mathbf{p}_k:\Lambda^k(T^*M)\otimes\ml{E}\mapsto M$ for every $0\leq k\leq n$. For every $t\in\IR$, one 
has $\mathbf{p}_k\circ \Phi_k^t=\varphi^t\circ\mathbf{p}_k$. 
If we define the corresponding Lie derivative on $\ml{E}$
\begin{equation}\label{e:lie}
\ml{L}_{V,\nabla}^{(k)}=\left(d^{\nabla}+\iota_V\right)^2:\Omega^k(M,\ml{E})\rightarrow \Omega^k(M,\ml{E}),
\end{equation}
where $\iota_V$ is the contraction by the vector field $V$, then, for any section $\psi_0$ in $\Omega^k(M,\ml{E})$ and for any $t\geq 0$, 
$\Phi_k^{-t*}(\psi_0)$ is the solution of the following partial differential equation:
\begin{equation}\label{e:heat-Lie-derivative}
\partial_t\psi=-\ml{L}_{V,\nabla}^{(k)}\psi,\ \quad\psi(t=0)=\psi_0\in\Omega^k(M,\ml{E}).
\end{equation}
As a tool to describe the long time behaviour of this equation, it is natural to form the following function
\begin{equation}\label{e:correlation}
\forall (\psi_1,\psi_2)\in\Omega^{n-k}(M,\ml{E}')\times\Omega^k(M,\ml{E}),\ C_{\psi_1,\psi_2}(t):=\int_M \psi_1\wedge\Phi_k^{-t*}(\psi_2),
\end{equation}
which we will call the \emph{correlation function}.

\subsection{Pollicott-Ruelle resonances}

Set $\mathcal{E}^\prime$ to be the dual bundle to $\cE$. Following the works of Pollicott~\cite{Po85} and Ruelle~\cite{Ru87a}, it is sometimes simpler to consider first the Laplace transform of this quantity,
\begin{equation}\label{e:correlation-laplace}
\forall (\psi_1,\psi_2)\in\Omega^{n-k}(M,\ml{E}')\times\Omega^k(M,\ml{E}),\ \hat{C}_{\psi_1,\psi_2}(z):=
\int_{0}^{+\infty}e^{-zt}\left(\int_M \psi_1 \wedge \Phi_k^{-t*}(\psi_2)\right)dt,
\end{equation}
which is well defined for $\text{Re}(z)>0$ large enough. In the above references, Pollicott and Ruelle showed that, for Axiom $A$ 
vector fields~\cite{Sm67} and for $\psi_1$ and $\psi_2$ 
compactly supported near the basic sets of the flow, this function admits a meromorphic extension to some 
half plane $\text{Re}(z)>-\delta$ with $\delta>0$ small enough. The poles and 
residues of this meromorphic continuation describe in some sense the fine structure of the flow 
$\Phi_k^{-t*}=e^{-t\ml{L}_{V,\nabla}^{(k)}}$ as $t\rightarrow+\infty$. More recently, it was proved by 
Butterley and Liverani that, for Anosov vector fields, the Laplace transformed correlators have meromorphic 
extensions to the entire complex 
plane without any restrictions on the supports of $\psi_1$ and $\psi_2$~\cite{BuLi07}. A different proof based on microlocal 
techniques was given by 
Faure and Sj\"ostrand~\cite{FaSj11} -- see also~\cite{DyZw13, DyGu14} for related 
results\footnote{We refer to section~\ref{s:spectral} for a more detailed account.}. In~\cite{DaRi16, DaRi17a}, we also proved this meromorphic extension in the case 
of Morse-Smale vector fields which are $\ml{C}^1$-linearizable -- see appendix~\ref{a:flows} for the precise definition.

Fix now $0\leq k\leq n$. To summarize, these recent developments showed that, for vector fields $V$ which are 
either Anosov or ($\ml{C}^1$-linearizable) Morse-Smale, there exists a minimal discrete\footnote{We mean that it has no accumulation point. In particular, it is at most countable.} set 
$\ml{R}_k(V,\nabla)\subset\IC$ such that, given any $(\psi_1,\psi_2)\in\Omega^{n-k}(M,\ml{E}')\times\Omega^{k}(M,\ml{E})$, the map $z\mapsto \hat{C}_{\psi_1,\psi_2}(z)$
 has a meromorphic extension whose poles are contained inside $\ml{R}_k(V,\nabla)$. These poles are called the \emph{Pollicott-Ruelle resonances}. Moreover, given any such 
 $z_0\in\ml{R}_k(V,\nabla)$, there exists an integer $m_k(z_0)$ and a linear map of \emph{finite rank}
\begin{equation}\label{e:spectral-projector}\pi_{z_0}^{(k)}:\Omega^k(M,\ml{E})\rightarrow\ml{D}^{\prime k}(M,\ml{E})\end{equation}
 such that, given any $(\psi_1,\psi_2)\in\Omega^{n-k}(M,\ml{E}')\times\Omega^{k}(M,\ml{E})$, one has, in a small neighborhood of $z_0$,
 $$\hat{C}_{\psi_1,\psi_2}(z)=\sum_{l=1}^{m_k(z_0)}(-1)^{l-1}\frac{\left\la(\ml{L}_{V,\nabla}^{(k)}+z_0)^{l-1}\pi_{z_0}^{(k)}(\psi_2),\psi_1\right\ra}{(z-z_0)^l}+R_{\psi_1,\psi_2}(z),$$
 where $R_{\psi_1,\psi_2}(z)$ is holomorphic. 
 Here, we use the convention that $\ml{D}^{\prime k}(M,\ml{E})$ represents the currents of degree $k$ with values in $\ml{E}$. 
 Elements inside the range of $\pi_{z_0}^{(k)}$ are called \emph{Pollicott-Ruelle resonant states} and, as we shall explain it below, they can be interpreted as the generalized eigenvectors 
 of the operator $-\ml{L}_{V,\nabla}^{(k)}$ acting on some appropriate Sobolev space. For any resonance $z_0$, we denote by $C_{V,\nabla}^k(z_0)$ the range of the 
 operator $\pi_{z_0}^{(k)}$ which is a finite dimensional space.

\subsection{Realization of De Rham cohomology via Pollicott-Ruelle resonant states}

 From the spectral interpretation of $C^{\bullet}_{V,\nabla}(z_0)$ -- see section~\ref{s:spectral} for details, one can deduce that, 
 for any\footnote{Whenever $z_0$ is not a resonance, we use the convention $C^k_{V,\nabla}(z_0)=\{0\}$.} $z_0\in\IC$,
\begin{equation}\label{e:ruelle-complex}0\xrightarrow{d^{\nabla}} C_{V,\nabla}^0(z_0)\xrightarrow{d^{\nabla}} C_{V,\nabla}^1(z_0)\xrightarrow{d^{\nabla}} \ldots\xrightarrow{d^{\nabla}} 
C_{V,\nabla}^n(z_0)\xrightarrow{d^{\nabla}} 0\end{equation}
defines a cohomological complex. We will denote by $\mathbf{H}^k(C_{V,\nabla}^{\bullet}(z_0),d^{\nabla})$ the cohomology of that complex. Our first main result states that this complex 
is isomorphic to the twisted De Rham complex when $z_0=0$:
\\
\\
\fbox{
\begin{minipage}{0.94\textwidth}  
\begin{theo}\label{t:quasiisomorphism} Let $\ml{E}\rightarrow M$ be a smooth complex vector bundle of dimension $N$ which is endowed 
with a flat connection $\nabla$. Suppose that $V$ is a vector field which is either Morse-Smale and $\ml{C}^1$-linearizable or Anosov. Then, for every 
$ 0\leq k\leq n$, the maps
$$\pi_{0}^{(k)}:\Omega^k(M,\ml{E})\rightarrow C_{V,\nabla}^k(0)$$
induce isomorphisms between $\mathbf{H}^k(M,\ml{E})$ and $\mathbf{H}^k(C_{V,\nabla}^{\bullet}(0),d^{\nabla})$.
\end{theo}
\end{minipage}
}
\\
\\

Observe that a direct corollary of this result is that the twisted De Rham cohomology is finite dimensional. Its dimension is usually denoted by 
$b_k(M,\ml{E})$ and it is called the $k$-th Betti number of $(\ml{E},\nabla)$. We will also explain in section~\ref{s:poincare} how to recover 
from our analysis that $b_k(M,\ml{E})=b_{n-k}(M,\ml{E}')$ for any $0\leq k\leq n$. In the statement of our Theorem, the assumption that the 
vector field is either Anosov or Morse-Smale comes from the fact that we have a good spectral theory for such systems 
thanks to~\cite{BuLi07, FaSj11, DyZw13, DaRi17a}. Yet, as we shall see, our proof is rather robust and this result should hold for any type of vector 
fields which has a good spectral behaviour yielding in particular an isolated eigenvalue at $z=0$ of finite multiplicity. Thanks to a classical 
result of Peixoto~\cite{Pe62} and to the Sternberg-Chen Theorem~\cite{Ne69, WWL08}, the assumption of being a $\ml{C}^1$-linearizable 
Morse-Smale vector field is in fact generic when $\text{dim} (M)=2$ -- see appendix~\ref{a:flows}.

In the case of Morse-Smale 
\emph{gradient} flows, this theorem was proved in~\cite{DaRi16} under the assumption that $\ml{E}$ is the trivial bundle $M\times\IC$ and it
provided a new (spectral) proof of earlier results due to Laudenbach~\cite{Lau92} and Harvey-Lawson~\cite{HaLa00, HaLa01}. For other flows, 
this result seems to be new and it gives a \emph{realization of the twisted De Rham cohomology as the cohomology of Pollicott-Ruelle 
resonant states}. Recall that in the case of Morse-Smale flows~\cite{DaRi17b}, we described explicitely a family of twisted 
currents generating the Pollicott-Ruelle resonant states and whose support are given by the stable manifolds. In a work related 
to the Anosov case~\cite{DyZw16}, Dyatlov and Zworski showed how to identify 
the Pollicott-Ruelle resonant states $u\in\ml{D}^{\prime k}(M,\IC)$ satisfying $\iota_V(u)=0$ with the De Rham 
cohomology of $X$ when $V$ is the geodesic vector field on $M:=SX$ (with $X$ negatively curved and $\text{dim}(X)=2$). Finally, in the case of Hodge-De Rham 
theory~\cite{BDIP96, Mn14}, the cohomology of the De Rham complex can be represented by the 
kernel of the Laplace operator $(d^{\nabla}+(d^{\nabla})^*)^2$, which is also finite dimensional. 

\subsection{Generalized Morse inequalities}

The gain compared with Hodge theory is that we do not have to go through Witten's deformation procedure to derive the Morse inequalities and that 
we can deal directly with the limit operator $-\ml{L}_{V,\nabla}$. This is of course 
at the expense of the spectral analysis of $-\ml{L}_{V,\nabla}$. Let us now draw some 
generalizations of the Morse-Smale inequalities from this first result. 
By standard arguments on exact sequences -- see 
e.g.~\cite[Par.~8.3]{DaRi16} for a brief reminder, one can verify that the following holds:
\\
\\
\fbox{
\begin{minipage}{0.94\textwidth}  
\begin{coro}[Spectral Morse inequalities]\label{c:morse-spectral}
 Suppose the assumptions of Theorem~\ref{t:quasiisomorphism} are satisfied. Then, 
the following holds:
$$\forall 0\leq k\leq n,\ \sum_{j=0}^k(-1)^{k-j}\operatorname{dim}C_{V,\nabla}^j(0)\geq \sum_{j=0}^k(-1)^{k-j}b_j(M,\ml{E}),$$
with equality in the case $k=n$.
\end{coro}
\end{minipage}
}
\\
\\
Recall that, for $k=n$, one recovers the Euler characteristic $\chi(M,\ml{E})$. If $V$ is a Morse-Smale gradient vector field (which is 
$\ml{C}^1$-linearizable) and if $\ml{E}=M\times\IC$ is the trivial bundle, we have shown in~\cite{DaRi16} that elements inside 
$C_{V,\nabla}^j(0)$ are generated by the currents of integration on unstable manifolds associated with critical points of index $k$. In particular, 
the dimension of that space is equal to the number of critical points of index $j$. Hence, this corollary applied to Morse-Smale gradient flows 
recovers in that case the classical Morse inequalities. Let us now apply this result in the more general framework of Morse-Smale vector fields. 
By definition~\cite{Sm60}, the nonwandering set of a Morse-Smale flow is composed of finitely many critical points and closed orbits, 
each of them being hyperbolic. Closed orbits $\Lambda$ can be divided into two categories: untwisted (if the corresponding stable manifold
\footnote{See the appendix of~\cite{DaRi17a} for a brief reminder.} $W^s(\Lambda)$ is 
orientable) and twisted (otherwise). We define the twisting index $\Delta_{\Lambda}$ of a closed orbit $\Lambda$ to be equal to $+1$ 
in the untwisted case and to $-1$ in the twisted case. Suppose now for simplicity of exposition that $\nabla$ preserves a 
Hermitian structure on $\ml{E}$ i.e. parallel transport by 
the flat connection $\nabla$ preserves the fiber metric of $\cE$. It implies that, for every closed orbit, the monodromy matrix $M_{\ml{E}}(\Lambda)$ for the parallel transport 
is a unitary matrix. In particular, its spectrum is included in $\IS^1$. For every closed orbit $\Lambda$, we then define $m_{\Lambda}$ to 
be the multiplicity of $\Delta_{\Lambda}$ as an eigenvalue of $M_{\ml{E}}(\Lambda)$ (this multiplicity can vanish and depends on $\nabla$). 
Using these conventions, we can state the following~:
 
\begin{coro}[Generalized Morse-Smale inequalities]\label{c:morse-smale-inequality} 
Let $\ml{E}\rightarrow M$ be a smooth, complex, 
hermitian 
vector bundle of rank $N$ endowed with a flat unitary connection $\nabla$. 
Suppose that $V$ is a Morse-Smale vector field. Denote by $c_k(V)$ the number of critical points of index\footnote{This means that the 
stable manifold of the point is of dimension $k$.} $k$. Then, the following holds, for every $0\leq k\leq n$,
$$\boxed{\sum_{\Lambda:\operatorname{dim}W^s(\Lambda)=k+1}m_{\Lambda}+N\sum_{j=0}^k(-1)^{k-j}c_k(V)\geq \sum_{j=0}^k(-1)^{k-j}b_j(M,\ml{E}),}$$
with equality in the case $k=n$.
\end{coro}

We will give the proof of this result in paragraph~\ref{ss:morse} and the main additional ingredient compared with Corollary~\ref{c:morse-spectral} 
is that we make use of the complete description of the Pollicott-Ruelle resonances for Morse-Smale flows that we gave in~\cite{DaRi17b}.  The fact 
that we require our flat connection to preserve a Hermitian 
structure is not optimal but makes the statement simpler to state. Recall in fact that our description of the fine structure 
of Pollicott-Ruelle resonances holds under weaker assumptions -- see~\cite[Sec.~4]{DaRi17b} for related discussion. We emphasize 
that this assumption is often made in Hodge theory~\cite{BDIP96, Mn14} especially 
for problems related to analytic torsion even if it is not completely necessary there too~\cite{Mu93}. Compared 
to Theorem~\ref{t:quasiisomorphism} and Corollary~\ref{c:morse-spectral}, 
we do not need to make the linearization assumption in that statement.

Again, when $V$ is a gradient flow, we recover the classical Morse inequalities. If 
$\ml{E}=M\times\IC$ is the trivial bundle, then $m_{\Lambda}=1$ for untwisted orbits and $m_{\Lambda}=0$ otherwise. In other words, 
the sum over the closed orbits is exactly the number of untwisted closed orbits whose stable manifold 
has dimension $k+1$. This inequality is slightly stronger than the original result of Smale~\cite{Sm60} whose upper bound involved also the number 
of twisted closed orbits. This stronger version of Morse-Smale inequalities was in fact proved by Franks in~\cite[Ch.~8]{Fr82} by a completely different approach than ours. 
In the case of more general vector bundles, we obtain Morse type inequalities which seem to be new.

\subsection{Koszul homological complexes}

As for the case of the coboundary operator, we can verify from the spectral definition of the spaces $ C_{V,\nabla}^\bullet(0)$ that
the sequence of linear maps~:
\begin{equation}\label{e:ruelle-koszul-complex-intro}0\xrightarrow{\iota_V} C_{V,\nabla}^n(0)\xrightarrow{\iota_V} C_{V,\nabla}^{n-1}(0)\xrightarrow{\iota_V} \ldots\xrightarrow{\iota_V} 
C_{V,\nabla}^0(0)\xrightarrow{\iota_V} 0,\end{equation}
defines a complex. We shall refer to this complex as the \textbf{Morse-Smale-Koszul} complex. In~\cite{DaRi16}, we showed that, for 
Morse-Smale gradient flows and for the trivial vector bundle $M\times \IC$, the Euler characteristic of the manifold is equal to the Euler 
characteristic of the Morse-Koszul complex. The same remains in fact true for Morse-Smale flows~:
\\
\\
\fbox{
\begin{minipage}{0.94\textwidth}  
\begin{theo}\label{t:koszul} 
Let $\ml{E}\rightarrow M$ be a smooth, complex, 
hermitian 
vector bundle of rank $N$ endowed with a flat unitary connection $\nabla$. 
Suppose that $V$ is a vector field which is Morse-Smale and $\ml{C}^{\infty}$-diagonalizable.\\
Then, the Euler characteristic $\chi(M,\ml{E})$ of $\ml{E}$ is equal to the Euler characteristic of the \textbf{Morse-Smale-Koszul} complex 
$(C_{V,\nabla}^\bullet(0),\iota_V)$. 
\end{theo}
\end{minipage}
}
\\
\\
We will prove this Theorem in section~\ref{s:koszul} and we will in fact show something slightly stronger as we will compute exactly the homology 
of that complex in each degree in terms of the critical points of $V$. The assumption of being $\ml{C}^{\infty}$ diagonalizable is defined in 
appendix~\ref{a:flows} and it holds as soon as certain (generic) nonresonance assumptions are satisfied by the Lyapunov exponents. This 
hypothesis appeared in our previous work~\cite{DaRi17b} in order to simplify the exposition and it could probably be 
removed if we are only interested in the part of the spectrum corresponding to resonances having a small enough real part. 
Yet, this would be at the expense of some extra (probably technical) work that was beyond the scope of~\cite{DaRi17b} and of the present article.

\subsection{Resonances on the imaginary axis and Reidemeister torsion}
\label{ss:mainresultstorsion}

We conclude by relating the Pollicott-Ruelle resonances to the results of Fried in~\cite[p.~44--53]{F87}. Recall that, in this reference, Fried introduced 
the \emph{twisted Fuller measure} of $(V,\nabla)$~:
\begin{def1}[Twisted Fuller measure]
In the notations of the 
previous paragraphs, the \emph{twisted Fuller measure} 
$ \mu_{V,\nabla}$ is a distribution in $\ml{D}'(\IR_+^*)$ defined by
the formula~:
\begin{eqnarray*}\
 \mu_{V,\nabla}(t) &:= &-\frac{N}{t}\sum_{\Lambda\ \text{fixed point}}\frac{\text{det}\left(\text{Id}-d_\Lambda\varphi^t\right)}
 {\left|\text{det}\left(\text{Id}-d_\Lambda\varphi^t\right)\right|}\\
 & +&  \frac{1}{t}
 \sum_{\Lambda\ \text{closed orbit}}\ml{P}_{\Lambda}\sum_{m\geq 1}\frac{\text{det}\left(\text{Id}-P_\Lambda^m\right)}
 {\left|\text{det}\left(\text{Id}-P_\Lambda^m\right)\right|}\text{Tr}\left(M_{\ml{E}}(\Lambda)^m\right)\delta(t-m\ml{P}_{\Lambda}),
 \end{eqnarray*}
where $P_{\Lambda}$ is a linearized Poincar\'e map associated with the closed orbit $\Lambda$ and where $\ml{P}_{\Lambda}$ is the minimal period of 
$\Lambda$. 
\end{def1}
In fact, Fried only defined this quantity when there are no critical points and 
he observed~\cite[p.~49]{F87} that, in this case, 
$t\mu_{V,\nabla}(t)$ coincides with the distributional traces of Guillemin and Sternberg~\cite[Ch.~VI]{GuSt90} -- see also~\cite{DyZw13} for a brief and 
self-contained account. We refer to the recent survey of Gou\"ezel~\cite[Sect.~2.4]{Go15} for a brief account on the role of these 
distributional traces in the study of transfer operators and Ruelle zeta functions~\cite{Ru76}. Here, we adopt a slightly more general definition than Fried's one which 
encompasses the case of critical points. If there are no closed orbits, then our definition also coincides with the distributional traces of~\cite{GuSt90}. In the mixed case, the formula would be slightly more subtle to justify 
from the functional perspective as we would have to deal with vector bundles having some ``jumps''.
Recall 
we denote by $C_{V,\nabla}^k(z_0)$ the range of the spectral
projector $\pi^{(k)}_{z_0}$ on the eigenspace of eigenvalue $z_0$ acting
on $k$--forms.
Our last Theorem relates these twisted Fuller measures 
with the correlation spectrum on the imaginary axis $i\IR$:
\\
\\
\fbox{
\begin{minipage}{0.94\textwidth}
\begin{theo}[Spectral interpretation of Fuller measures]
\label{t:fried0}
Let $\ml{E}\rightarrow M$ be a smooth, complex, 
hermitian 
vector bundle of rank $N$ endowed with a flat unitary connection $\nabla$. 
Suppose that $V$ is a vector field which is Morse-Smale and $\ml{C}^{\infty}$-diagonalizable.

Then, one has
$$\mu_{V,\nabla}(t)=\sum_{k=0}^n(-1)^{n-k+1}\sum_{z_0\in\ml{R}_k(V,\nabla)\cap i\IR}\operatorname{dim} 
\left(C_{V,\nabla}^k(z_0)\cap\operatorname{Ker}(\iota_V)\right)\frac{e^{t z_0}}{t}$$
in the sense of distributions in $\IR_+^*$.
\end{theo}
\end{minipage}
}
\\
\\
We will prove this Theorem in section~\ref{s:fried}. 
In the particular case where $V$ is non singular, the Fuller measure can be related
to the Guillemin trace as follows.
We denote by $V^\perp\subset T^*M$ the vector bundle $\ker(\iota_V)$ of invariant 
$1$-forms by the flow and by 
$\Lambda^kV^\perp$ the bundle of invariant $k$-forms. 
The Guillemin trace on twisted invariant forms of 
degree $k$ is denoted by $TR^\flat_{\Lambda^kV^\perp\otimes\cE}$.
Then, following Fried~\cite[p.~49--50 in particular equation VII]{F87}, a direct application of 
the Guillemin trace formula yields an identity relating twisted Fuller measures, Guillemin traces and
imaginary resonances as follows~:
\begin{eqnarray*}
 t\mu_{V,\nabla}(t) &=&\sum_{k=0}^n(-1)^{n-k+1} TR^\flat_{\Lambda^kV^\perp\otimes\cE}\left(e^{-t\mathcal{L}_{V,\nabla}} \right)\\
&=& \sum_{k=0}^n(-1)^{n-k+1}\sum_{z_0\in\ml{R}_k(V,\nabla)\cap i\IR}\operatorname{dim} 
\left(C_{V,\nabla}^k(z_0)\cap\operatorname{Ker}(\iota_V)\right)e^{tz_0}
\end{eqnarray*}
in the sense of distributions in $\IR_+^*$.

Let us now briefly recall the topological content of these twisted 
Fuller measures (and thus of the Pollicott-Ruelle spectrum). 
Fried observed that the twisted Fuller measure shares a lot of 
properties with some topological invariants such as the Euler 
characteristic or the Reidemeister torsion~\cite[p.~46-49]{F87}. 
In fact, we can observe that 
the coefficients defining $\mu_{V,\nabla}(t)$ are of purely topological nature even if the support of the measure depends on the 
parametrization of the flow and hence is not topological. In order to get rid of these nontopological informations, Fried evaluated 
the ``total mass'' of these Fuller measures and showed that in certain cases this defines indeed a topological invariant. For that 
purpose, one can introduce the following zeta function~:
$$
\zeta_{V,\nabla}^{\flat}(s,z):=\frac{1}{\Gamma(s)}\int_{0}^{+\infty} e^{-tz}\mu_{V,\nabla}(t)t^{s}dt
$$
and following Fried~\cite[p.~51]{F87}, we define the corresponding \textbf{torsion function} 
$$Z_{V,\nabla}(z)=\exp\left(-\frac{d}{ds}\zeta^\flat_{V,\nabla}(s,z)|_{s=0}\right)$$
as the exponential of $-\partial_{s}\zeta^{\flat}_{V,\nabla}(0,z)$.
We emphasize that Fried directly used the Laplace transform to define his twisted zeta function as he only considered nonsingular vector fields. 
Here, due to the fact that we allow critical points, we have to make this extra regularization involving a Mellin transform. In our slightly 
more general framework, the terminology ``torsion function'' may sound a little bit abusive as we do not recover the Reidemeister torsion in general. Computing 
the derivative of $\zeta_{V,\nabla}^{\flat}(s,z)$ w.r.t. $s$ at $s=0$, 
one can verify (see paragraph~\ref{ss:hodge-torsion}) that, for $\text{Re}(z)$ large enough,
$$\partial_{s}\zeta^{\flat}_{V,\nabla}(0,z)=\chi(M,\ml{E})\ln z+\sum_{\Lambda\ \text{closed orbit}}(-1)^{\text{dim}\ W^u(\Lambda)}
\ln \text{det}\left(\text{Id}-e^{-\ml{P}_{\Lambda }z}\Delta_{\Lambda}M_{\ml{E}}(\Lambda)\right).$$
Therefore, the torsion function $Z_{V,\nabla}(z)$ satisfies the
identity~:
\begin{eqnarray*}
\boxed{ Z_{V,\nabla}(z)=z^{-\chi(M,\ml{E})} \prod_{\Lambda\ \text{closed orbit}}
\text{det}\left(\text{Id}-e^{-\ml{P}_{\Lambda }z}\Delta_{\Lambda}M_{\ml{E}}(\Lambda)\right)^{-(-1)^{\text{dim}\ W^u(\Lambda)}}}
\end{eqnarray*}
for $\text{Re}(z)$ large enough. The right hand side of the above equality is a weighted Ruelle zeta function~\cite[equation (13) p.~19]{Go15},
the weight being given by the \textbf{monodromy of the flat connection}. 
If the vector field is nonsingular 
which means the Euler characteristic term vanishes
and if $m_{\Lambda}=0$ for every closed orbit\footnote{In particular, $(\ml{E},\nabla)$ is acyclic.}, then $Z_{V,\nabla}$ converges as 
$z\rightarrow 0^+$ and Fried showed that the modulus of the limit is equal to the Reidemeister torsion~\cite[Sect.~3]{F87}.

Following the works of Ray--Singer on analytic torsion~\cite{RaySi71}, we now introduce the following:
\begin{def1}[Spectral zeta determinant of the nonzero resonances on the imaginary axis]
We define the \emph{spectral zeta function}:
$$
\zeta_{RS}(s,z):=\sum_{k=0}^n(-1)^{n-k}k\sum_{z_0\in\ml{R}_k(V,\nabla)\cap i\IR^*}\text{dim} \left(C_{V,\nabla}^k(z_0)\right)
\left(z-z_0\right)^{-s}
$$
and the corresponding spectral zeta determinant\footnote{This regularizes the infinite product
$ \prod_{z_0\in \ml{R}_k(V,\nabla)\cap i\IR^*} \left(z-z_0\right)^{(-1)^kk\text{dim} \left(C_{V,\nabla}^k(\lambda)\right)}  $
of nonzero resonances on the imaginary axis.} $$\boxed{Z_{RS}(z):=\exp\left(-\frac{d}{ds}\zeta_{RS}(s,z)|_{s=0}\right)}.$$

\end{def1}
Then, for $\text{Re}(z)$ large enough, as a consequence of Theorem~\ref{t:fried0} and of standard arguments from Hodge theory 
(see paragraph~\ref{ss:hodge-torsion}), 
both the flat zeta function $\zeta_{V,\nabla}^{\flat}(s,z)$ and 
the
\emph{spectral zeta function}
$ \zeta_{RS}(s,z)$ admit analytic continuations in $s\in \mathbb{C}$
and are related by the equation~:
\begin{eqnarray*}
\boxed{\zeta_{V,\nabla}^{\flat}(s,z)=-\left(\chi(M,\cE)+\sum_{\Lambda\ \text{closed orbit}}(-1)^{\text{dim}\ W^u(\Lambda)}m_{\Lambda}\right)z^{-s}+\zeta_{RS}(s,z).}
\end{eqnarray*}

Hence, as a direct corollary of our Theorem and of Fried's Theorem, we have:
\begin{coro}\label{c:Spectralzetaruelle}
Suppose that the assumptions of Theorem~\ref{t:fried0} are satisfied. Then, one has:
\begin{eqnarray*}
\boxed{Z_{RS}(z)=\prod_{\Lambda\ \text{closed orbit}}
\left(z^{-m_{\Lambda}}\operatorname{det}\left(\operatorname{Id}-e^{-\ml{P}_{\Lambda }z}\Delta_{\Lambda}M_{\ml{E}}(\Lambda)\right)\right)^{-(-1)^{\operatorname{dim}\ W^u(\Lambda)}}.}
\end{eqnarray*}

In particular if $V$ is non singular and $m_\Lambda=0$ for every periodic orbit $\Lambda$ then $\vert Z_{RS}(0) \vert$ coincides
with the \textbf{Reidemeister torsion}. 
\end{coro}
Informally, the above discussion means that we can recover the Euler characteristic and the Reidemeister torsion 
(under proper assumptions) from the Pollicott-Ruelle spectrum lying on the imaginary axis.
These different results illustrate that the Pollicott-Ruelle spectrum on the imaginary axis has a nice topological interpretation in the case of Morse-Smale flows. 
More specifically, the topological content seems to be contained in the first band of resonances in the terminology of 
Faure and Tsujii~\cite{FaTs13, FaTs17}. Recall that, for Anosov contact flows, they showed that 
the Pollicott-Ruelle spectrum exhibits a band structure. We proved in~\cite{DaRi17b} that the same band structure remains true 
in the case of Morse-Smale flows and that the first band of resonances is in that case contained inside the imaginary axis. 
Hence, the main observation of these results is that, at least in certain acyclic cases, 
\textbf{we can keep track of topological invariants inside the first band of the correlation spectrum and not only inside the kernel} 
which is here reduced to $\{0\}$.

In section~\ref{s:torsion}, we will be even more precise on the topological features of the first band of resonant states for 
Morse-Smale flows. In fact, for every 
$z_0\in i\IR^*$, $(C^\bullet_{V,\nabla}(z_0),d^{\nabla})$ 
defines an acyclic cohomological complex of finite dimensional spaces -- see paragraph~\ref{ss:chain-homotopy}. Thus, once we have 
specified a preferred basis of that space, we can compute the torsion of that complex~\cite{F87, Mn14} which is a certain determinant 
related to the coboundary operator. In the case of Morse-Smale flows, such a preferred basis of Pollicott-Ruelle resonant states can be naturally 
introduced following our previous work~\cite{DaRi17b} -- see paragraphs~\ref{sss:local-coordinates-critical-point} 
and~\ref{sss:local-coordinates-closed-orbit} for a brief reminder. Then, we will introduce the infinite dimensional space
$$C^{\bullet}_{V,\nabla}(i\IR^*)=\bigoplus_{z_0\in \ml{R}_{\bullet}(V,\nabla)\cap i\IR^*}C^\bullet_{V,\nabla}(z_0),$$
which still defines an acyclic complex. We will define a notion of regularized torsion of this infinite dimensional ``complex'' 
by regularizing the infinite product of torsions of 
every individual finite dimensional complex of Pollicott-Ruelle resonant states 
associated to the 
spectrum on the imaginary axis. 
In paragraph~\ref{s:fried}, we will see that this 
regularized torsion is related to the behaviour of the torsion function $Z_{V,\nabla}(z)$ at $z=0$ and to the spectral zeta function $\zeta_{RS}$. 
Thus, it is also linked to the Reidemeister torsion whenever 
$V$ has no critical points and where $m_{\Lambda}$ is equal to $0$ for every closed orbit.

\subsection{Conventions}
All along the article, $M$ will denote a smooth, compact, oriented, boundaryless manifold of dimension $n$. We shall denote by $V$ a smooth vector 
field on $M$ and by $\ml{E}\rightarrow M$ a complex vector bundle of dimension $N$ which is endowed with a flat connection $\nabla$. Except mention 
of the contrary, $\ml{E}$ is not necessarly endowed with a Hermitian structure compatible with $\nabla$.

\subsection*{Acknowledgements} In the spring 2016, Viviane Baladi asked us if there were some relations between our work~\cite{DaRi16} 
and the works of Fried on Reidemeister torsion: we warmly thank her for pointing to us this series of works which motivated part of the results 
presented here. We also thank Fr\'ed\'eric Faure for many explanations on his works with Johannes Sj\"ostrand and Masato Tsujii. We also acknowledge 
useful discussions related to this article and its companion articles~\cite{DaRi17a, DaRi17b} with Livio Flaminio, 
Colin Guillarmou, Benoit Merlet, Fr\'ed\'eric Naud and Patrick Popescu Pampu. The second author is partially supported 
by the Agence Nationale de la Recherche through the Labex CEMPI (ANR-11-LABX-0007-01) and the 
ANR project GERASIC (ANR-13-BS01-0007-01).

\section{Spectral theory of Anosov and Morse-Smale vector fields}\label{s:spectral}

We begin with a brief account on the existence of Pollicott-Ruelle resonant states following the microlocal approach of 
Faure and Sj\"ostrand~\cite{FaSj11} -- see also~\cite{DyZw13, DaRi16} for extensions of these results. In the Anosov case, another approach would be to 
use the earlier results of Butterley and Liverani via spaces of anisotropic H\"older 
distributions~\cite{BuLi07}. We could probably also proceed via coherent states like in the works of Faure and Tsujii~\cite{Ts12, FaTs17}. 
We also refer to~\cite{Ba16} for a recent account by Baladi concerning the related case of hyperbolic diffeomorphisms. 

\subsection{Escape function}

The key ingredient in Faure-Sj\"ostrand's analysis is the construction of a so-called \textbf{escape function}, or equivalently a Lyapunov function for the 
Hamiltonian flow on $T^*M$ induced by
$$\forall(x,\xi)\in T^*M,\quad H_V(x,\xi):=\xi(V(x)).$$
Recall that the corresponding Hamiltonian flow is given by
$$\forall t\in\IR,\quad\Phi_V^t(x,\xi):=\left(\varphi^t(x,\xi), \left(d\varphi^t(x)^T\right)^{-1}\xi\right),$$
where $\varphi^t$ is the flow induced by the vector field $V$ on $M$. We shall denote by $X_V(x,\xi)$ the Hamiltonian vector field induced by $\Phi_V^t$.
Using the terminology of~\cite{FaSj11}, an \textbf{escape function} for the flow $\Phi_V^t$ on 
$T^*M$ is a function of the form
$$G(x,\xi):=m(x,\xi)\log\sqrt{1+f(x,\xi)^2},$$
meeting the following requirements:
\begin{itemize}
 \item \textbf{$m$ is a symbol of order $0$}. This means $m(x,\xi)$ belongs to $\ml{C}^{\infty}(T^*M)$ and, for all multi-indices $(\alpha,\beta)\in\IN^{2n}$, 
 $$\partial^{\alpha}_x\partial_{\xi}^{\beta} m(x,\xi)=\ml{O}\left((1+\|\xi\|_x^2)^{-\frac{|\beta|}{2}}\right).$$
 \item \textbf{$f$ has the correct homogeneity in $\xi$}. $f(x,\xi)$ belongs to $\ml{C}^{\infty}(T^*M)$, and, for $\|\xi\|_x\geq 1$, $f(x,\xi)> 0$ is positively homegeneous of degree $1$.
 \item \textbf{The symbol $H_V$ is micro elliptic in some conical open set $\ml{N}_0$}. There exist a conical open set $\ml{N}_0$ and a constant $C>0$ such that, for every $(x,\xi)\in \ml{N}_0$ satisfying $\|\xi\|_x\geq C$, one has 
 $$|H_V(x,\xi)|\geq \frac{1}{C}(1+\|\xi\|_x),$$
 and $f(x,\xi)=|H_V(x,\xi)|$. Moreover, outside some slightly bigger conical open set $(\ml{N}_0\subset)\ml{N}_1$, one has $f(x,\xi)=\|\xi\|_x$ for $\|\xi\|_x\geq 1$.
 \item \textbf{Outside $\ml{N}_0$, $G$ has uniform decrease along the Hamiltonian flow.} There exist $R\geq 1$ and $c>0$ such that, for every $(x,\xi)\in T^*M$ satisfying $\|\xi\|_x\geq R$,
 $$ X_V.G(x,\xi)\leq 0\ \text{and}\ (x,\xi)\notin \ml{N}_0\Longrightarrow X_V.G(x,\xi)\leq -c.$$
\end{itemize}
In other words, the escape function is strictly decreasing except in the directions where the Hamiltonian $H_V$ is elliptic. Observe that $\ml{N}_0$ and $\ml{N}_1$ 
can be even chosen to be empty -- this was for instance the case for gradient flows~\cite{DaRi16}. Note that compared with Lemma~1.2 in~\cite{FaSj11}, we do not 
require anything on the value of the \textbf{order function} $m$ in certain directions of phase space. 
The reason for this is that the above axioms are the key ingredients to build a convenient spectral theory for the vector field $V$. The other requirements 
on $m$ are additional informations on the structure Sobolev space which are specific to the Anosov framework. According to~\cite{FaSj11, DaRi17a}, examples of vector fields possessing such an escape function are 
given by Anosov or Morse-Smale vector fields which are $\ml{C}^1$-linearizable. Motivated by the main results from~\cite{FaSj11} and in order to alleviate the 
notations in the upcoming sections, we shall say that a flow $\varphi^t$ possessing such an escape function is \textbf{microlocally tame}. We shall make this assumption in the following and we 
keep in mind that the only known examples of such flows are the ones mentionned above.

\subsection{Anisotropic Sobolev spaces}
Let us now briefly recall the construction of Faure-Sj\"ostrand based on the existence of such a function~\cite{FaSj11}. Note that, in this reference, the authors only treated the case 
of the trivial bundle $M\times\IC$. The extension of this microlocal approach to more general bundles was made by Dyatlov and Zworski in~\cite{DyZw13} by some slightly different point of view. 
Fix $s>0$ some large parameter 
(corresponding to the Sobolev regularity we shall require) and $0\leq k\leq n$ (the degree of the forms). We also fix some Riemannian metric $g$ on $M$ and some hermitian 
structure $\la.,\ra_{\ml{E}}$ on $\ml{E}$, none of them being a priori related to $V$ and $\nabla$. We can fix the inner product $\la,\ra_{g^*}^{(k)}$ on $\Lambda^k(T^*M)$ 
which is induced by the metric $g$ on $M$. Then, the
the Hodge star operator is the unique isomorphism 
$\star_k :\Lambda^k(T^*M)\rightarrow \Lambda^{n-k}(T^*M)$ such that, for every $\psi_1$ in $\Omega^k(M)\simeq\Omega^k(M,\IC)$ and $\psi_2$ in $\Omega^{n-k}(M)\simeq\Omega^{n-k}(M,\IC)$, 
$$\int_M\psi_1\wedge\psi_2=\int_M\la \psi_1,\star_k^{-1}\psi_2\ra_{g^*(x)}^{(k)}\omega_g(x),$$
where $\omega_g$ is the volume form induced by the Riemannian metric on $M$. This induces a map $\star_k:\Omega^k(M,\ml{E})\rightarrow \Omega^{n-k}(M,\ml{E})$ which acts trivially on the 
$\ml{E}$-coefficients. Using the metric $g_{\ml{E}}$ on $\ml{E}$, we can also introduce the following pairing, for every $0\leq k,l\leq n$,
$$\la.\wedge.\ra_{\ml{E}}:\Omega^k(M,\ml{E})\otimes_{\ml{C}^{\infty}(M)}\Omega^l(M,\ml{E})\rightarrow \Omega^{k+l}(M).$$
Combining these, we can define the positive definite Hodge inner product on $\Omega^{k}(M,\ml{E})$ as
$$(\psi_1,\psi_2)\in\Omega^k(M,\ml{E})\times\Omega^k(M,\ml{E})\mapsto\int_M\la \psi_1\wedge\star_k(\psi_2)\ra_{\ml{E}}.$$
In particular, we can define $L^2(M,\Lambda^k(T^*M)\otimes\ml{E})$ as the completion of $\Omega^k(M,\ml{E})$ for this scalar product.

Set now 
\begin{equation}\label{A_mweightsobo}
A_s^{(k)}(x,\xi):=\exp \left((s m(x,\xi)+n-k)\log \left(1+f(x,\xi)^2)^{\frac{1}{2}}\right)\right),
\end{equation}
where $G(x,\xi)=m(x,\xi)\log \left(1+f(x,\xi)^2)^{\frac{1}{2}}\right)$ is the escape function. We can define $$\mathbf{A}_s^{(k)}(x,\xi):=A_s^{(k)}(x,\xi)\textbf{Id}_{\Lambda^k(T^*M)\otimes\ml{E}}$$ 
belonging to $\text{Hom}(\Lambda^k(T^*M)\otimes\ml{E})$ 
and introduce an anisotropic Sobolev space of currents by setting
$$\ml{H}^{sm+n-k}_k(M,\ml{E})=\Op(\mathbf{A}_s^{(k)})^{-1}L^2(M,\Lambda^k(T^*M)\otimes\ml{E}),$$
where $\Op(\mathbf{A}_s^{(k)})$ is a (essentially selfadjoint) pseudodifferential operator with principal symbol $\mathbf{A}_s^{(k)}$. 

\begin{rema} Note that this requires to deal with symbols of variable order $m(x,\xi)$ 
whose symbolic calculus was described in Appendix~A of~\cite{FaRoSj08}. This can be done as the symbol $m(x,\xi)$ 
belongs to the standard class of symbols $S^0(T^*M)$. We also refer to~\cite[App.~C.1]{DyZw13} for a brief reminder of pseudodifferential operators 
with values in vector bundles. In particular, adapting the proof of~\cite[Cor.~4]{FaRoSj08} to the vector bundle valued framework, one can verify that $\mathbf{A}_s^{(k)}$ is an elliptic symbol, and 
thus $\Op(\mathbf{A}_s^{(k)})$ can be chosen to be invertible. 
\end{rema}

We observe from the composition rule for pseudodifferential operators that
\begin{equation}\label{e:exterior-derivative-sobolev}
 \forall 0\leq k\leq n-1,\quad d^{\nabla}:\ml{H}^{sm+n-k}_k(M,\ml{E})\rightarrow \ml{H}^{sm+n-(k+1)}_{k+1}(M,\ml{E}).
\end{equation}
Mimicking the proofs of~\cite{FaRoSj08}, we can also deduce some properties of these spaces of currents. First of all, they are endowed with a Hilbert 
structure inherited from the $L^2$-structure on $M$. The space $$\ml{H}^{sm+n-k}_k(M,\ml{E})^{\prime}=\Op(\mathbf{A}_s^{(k)})L^2(M,\Lambda^k(T^*M)\otimes\ml{E})$$
is the topological dual of $\ml{H}^m_k(M,\ml{E})$ which is in fact reflexive. The following injections holds and they are continuous:
$$\Omega^k(M,\ml{E}) \subset \ml{H}_k^{m}(M,\ml{E}) \subset \mathcal{D}^{\prime,k}(M,\ml{E}),$$ 
The Riemannian metric on $\ml{E}$ allows to define a canonical isomorphism $\tau_{\ml{E}}:\ml{E}\rightarrow\ml{E}'$ by setting, for every $\psi$ in $\ml{E}$, $\tau_{\ml{E}}(\psi)=\la\psi,.\ra_{\ml{E}}$. 
Then, combined with the Hodge star map, it induces an isomorphism from 
$\ml{H}^{m}_k(M,\ml{E})^{\prime}$ to $\ml{H}^{-m}_{n-k}(M,\ml{E}')$, 
whose Hilbert structure is given by the scalar product
$$(\psi_1,\psi_2)\in\ml{H}^{-m}_{n-k}(M,\ml{E}')^2\mapsto \la \star_k^{-1}\tau_{\ml{E}}^{-1}(\psi_1),\star_k^{-1}\tau_{\ml{E}}^{-1}(\psi_2)\ra_{\ml{H}_k^m(M,\ml{E})'}.$$
Thus, the topological dual of $\ml{H}_k^m(M,\ml{E})$ can be identified with $\ml{H}^{-m}_{n-k}(M,\ml{E}')$, where, for every $\psi_1$ in $\Omega^k(M,\ml{E})$ 
and $\psi_2$ in $\Omega^{n-k}(M,\ml{E}')$, one has the following duality relation:
\begin{eqnarray*}\la\psi_2,\psi_1\ra_{\ml{H}_{n-k}^{-m}(M,\ml{E}')\times\ml{H}_k^m(M,\ml{E})} & = & \int_{M}\psi_2\wedge \psi_1\\
& = &\la \Op(\mathbf{A}_m^{(k)})^{-1}\star_k^{-1}\tau_{\ml{E}}^{-1}(\overline{\psi_2}),\Op(\mathbf{A}_m^{(k)})\psi_1\ra_{L^2(M,\Lambda^k(T^*M)\otimes\ml{E})}\\
& = &\la\star_k^{-1}\tau_{\ml{E}}^{-1}(\psi_2),\psi_1\ra_{\ml{H}_k^m(M,\ml{E})\times\ml{H}_{k}^m(M,\ml{E})^{\prime}}.\end{eqnarray*}

\subsection{Pollicott-Ruelle resonances} Now that we have settled the functional framework for our microlocally tame vector field, one can establish 
the existence of Pollicott-Ruelle resonances. More precisely, one can show that, for every $T_0>0$, there exists some $s>0$ large enough such that, for every $0\leq k\leq n$,
$$-\ml{L}_{V,\nabla}^{(k)}:\ml{H}^{sm+n-k}_k(M,\ml{E})\rightarrow\ml{H}^{sm+n-k}_k(M,\ml{E})$$
has a \textbf{discrete spectrum} on the half plane $\text{Im}(z)>-T_0$, consisting of eigenvalues of finite algebraic multiplicity.
We refer to~\cite[Th.~1.4]{FaSj11} for a more precise statement. This result essentially follows from the fact that Faure and Sj\"ostrand proved that 
$(\ml{L}_{V,\nabla}+z)$ is a Fredholm depending analytically on $z$. The eigenvalues are the so-called \textbf{Pollicott-Ruelle resonances} and 
the corresponding generalized eigenmodes are called \textbf{resonant states.} These objects correspond exactly to the poles and the residues of the 
function $\hat{C}_{\psi_1,\psi_2}(z)$ defined in~\eqref{e:correlation-laplace}.

\begin{rema}
In~\cite[Sect.~3]{FaSj11}, the authors only treated the case $k=0$ and $\ml{E}=M\times \IC$. Yet, their proof can be adapted almost
verbatim to our framework except that we have to deal with pseudodifferential operators with values in $\Lambda^k(T^*M)\otimes\ml{E}$ -- 
see~\cite[Part~I.3]{BDIP96}~\cite[App.~C]{DyZw13} for a brief review. As was already observed 
in~\cite{DyZw13}, the main point to extend directly Faure-Sj\"ostrand's proof to more general bundles is that 
the (pseudodifferential) operators under consideration have a \emph{diagonal symbol}. In fact, given any local basis $(e_j)_{j=1,\ldots J_k}$ 
of $\Lambda^k(T^*M)\otimes\ml{E}$ and any family $(u_j)_{j=1,\ldots J_k}$ of smooth functions $\ml{C}^{\infty}(M)$, one has
$$\ml{L}_{V,\nabla}^{(k)}\left(\sum_{j=1}^{J_k}u_je_j\right)=\sum_{j=1}^{J_k}\ml{L}_{V}(u_j)e_j+\sum_{j=1}^{J_k}\ml{L}_{V,\nabla}^{(k)}(e_j)u_j,$$
where the second part of the sum in the right-hand side is a lower order term (of order $0$). In other words, the principal symbol of 
$\ml{L}_{V,\nabla}^{(k)}$ is $\xi(V(x))\mathbf{Id}_{\Lambda^k(T^*M)\otimes\ml{E}}$. This diagonal form allows to adapt the 
proofs of~\cite{FaSj11} to this vector bundle framework.
\end{rema}

We conclude this preliminary paragraph with some properties of this spectrum:
\begin{itemize}
 \item Faure and Sj\"ostrand implicitely show~\cite[Lemma~3.3]{FaSj11} that, 
for every $z$ in $\mathbb{C}$ satisfying $\text{Im} z>C_0$ (for some $C_0>0$ large enough), one has
\begin{equation}\label{e:norm-resolvent}
 \left\|\left(\ml{L}_{V,\nabla}^{(k)}+z\right)^{-1}\right\|_{\ml{H}^m_k(M,\ml{E})\rightarrow\ml{H}_k^m(M,\ml{E})}\leq\frac{1}{\text{Re}(z)-C_0}.
\end{equation}
In particular, combining with the Hille-Yosida Theorem~\cite[Cor.~3.6, p.~76]{EnNa00}, it implies that
\begin{equation}\label{e:semigroup}e^{-t\ml{L}_{V,\nabla}^{(k)}}=(\Phi_k^{-t})^*:\ml{H}^m_k(M,\ml{E})\rightarrow \ml{H}_k^m(M,\ml{E}),\end{equation}
generates a strongly continuous semigroup which is defined for every $t\geq 0$ and whose norm is bounded by $e^{tC_0}.$
 \item As in~\cite[Th.~1.5]{FaSj11}, we can show that the eigenvalues (counted with their algebraic multiplicity) 
and the eigenspaces 
of $-\ml{L}_{V,\nabla}^{(k)}:\ml{H}_k^{sm+n-k}(M,\ml{E})\rightarrow \ml{H}_k^{sm+n-k}(M,\ml{E})$ are in fact independent of the choice of escape function and of the parameter $s$. 
 \item  By duality, the same spectral properties holds for the dual operator
\begin{equation}\label{e:dual}(-\ml{L}_{V,\nabla}^{(k)})^{\dagger}=-\ml{L}_{-V,\nabla^{\dagger}}^{(n-k)}:\ml{H}^{-m}_{n-k}(M,\ml{E}')\rightarrow\ml{H}^{-m}_{n-k}(M,\ml{E}'),\end{equation}
with $\nabla^{\dagger}$ the dual connection~\cite[Sect.~5]{DaRi17a}.
 \item Given any $z_0$ in $\IC$, the corresponding spectral projector $\pi_{\lambda}^{(k)}$ is given by~\cite[Appendix]{HeSj86}:
\begin{equation}\label{e:spectral-proj}\pi_{z_0}^{(k)}:=\frac{1}{2i\pi}\int_{\gamma_{z_0}}(z+\ml{L}_{V,\nabla}^{(k)})^{-1}dz:\ml{H}_k^{sm+n-k}(M,\ml{E})\rightarrow \ml{H}_k^{sm+n-k}(M,\ml{E}),\end{equation}
 where $\gamma_{z_0}$ is a small contour around $z_0$ which contains at most the eigenvalue $z_0$ in its interior.
 \item Given any $z_0$ in $\IC$ with $\operatorname{Re} (\lambda)>-T_0$, there exists $m_k(z_0)\geq 1$ such that, in a small neighborhood of $z_0$, one has
 \begin{equation}\label{e:resolvent}
  (z+\ml{L}_{V,\nabla}^{(k)})^{-1}=\sum_{j=1}^{m_k(z_0)}(-1)^{j-1}\frac{(\ml{L}_{V,\nabla}^{(k)}+z_0)^{j-1}\pi_{z_0}^{(k)}}{(z-z_0)^j}+R_{z_0,k}(z):\ml{H}_k^{sm+n-k}(M,\ml{E})\rightarrow \ml{H}_k^{sm+n-k}(M,\ml{E}),
 \end{equation}
with $R_{z_0,k}(z)$ an holomorphic function.
\end{itemize}
The last two statements allow to make the connection between this spectral framework and the functions $\hat{C}_{\psi_1,\psi_2}(z)$ appearing in the introduction.

\section{Morse-Smale complex and generalized Morse inequalities}\label{s:proof-isomorphism}

In this section, we will give the proof of Theorem~\ref{t:quasiisomorphism} and draw some consequences in terms of 
Morse-Smale inequalities and of spectral trace asymptotics. In all this section, $V$ is a smooth vector field which 
is microlocally tame in the sense of section~\ref{s:spectral} and $\nabla$ is a flat connection on $\ml{E}$. 

\begin{rema}
Except mention of the contrary, we will not suppose that $\nabla$ preserves some hermitian structure on $\ml{E}$ (or equivalently that $\ml{E}$ is endowed 
with a flat unitary connection $\nabla$). 
In that sense, it is slightly different from the case of Hodge theory~\cite{BDIP96, Mn14}. Recall that preserving a hermitian structure would mean that, for any $(\psi_1,\psi_2)$ in 
$\Omega^{0}(M,\ml{E})\times\Omega^{0}(M,\ml{E})$, one has
$$d(\la \psi_1,\psi_2\ra_{\ml{E}})=\la \nabla \psi_1,\psi_2\ra_{\ml{E}}+\la  \psi_1,\nabla\psi_2\ra_{\ml{E}}.$$
This last property ensures that, for any $\psi_1$ in $\Omega^{k-1}(M)$ and for any $\psi_2$ in $\Omega^{n-k}(M)$,
\begin{equation}\label{e:duality-coboundary-2}\int_M\la d^{\nabla}\psi_1,\psi_2\ra_{\ml{E}}=(-1)^k\int_{M}\la \psi_1,d^{\nabla}\psi_2\ra_{\ml{E}}.\end{equation}
It also imposes that the monodromy matrices for the parallel transport are unitary -- see~\cite[Sect.~4]{DaRi17b}.
\end{rema}

\subsection{De Rham cohomology}\label{ss:cohomology}

We start with a brief reminder on de Rham cohomology~\cite{dRh80, BDIP96}. An element $\omega$ in $\Omega^*(M,\ml{E})$ 
such that $d^{\nabla}\omega=0$ is called a \textbf{cocycle} while an 
element $\omega$ which is equal to $d^{\nabla}\alpha$ for some $\alpha\in \Omega^*(M,\ml{E})$ is called a \textbf{coboundary}. We define then
$$Z^k(M,\ml{E})=\text{Ker}(d^{\nabla}) \cap\Omega^k(M,\ml{E}),\ \text{and}\ B^k(M,\ml{E})=\text{Ran}(d^{\nabla}) \cap\Omega^k(M,\ml{E}).$$
Obviously, $B^k(M,\ml{E})\subset Z^k(M,\ml{E})$, and the quotient space $\mathbf{H}^k(M,\ml{E})= Z^k(M,\ml{E})/B^k(M,\ml{E})$ is called the 
\textbf{$k$-th de Rham cohomology}. The complex is said to be \textbf{acyclic} if all the cohomology is reduced to $\{0\}$.

The coboundary operator $d^{\nabla}$ can be extended into a map acting 
on the space of currents $\ml{D}^{\prime,*}(M,\ml{E})$. This allows to define another 
cohomological complex $(\ml{D}^{\prime,*}(M,\ml{E}),d^{\nabla})$:
$$0\xrightarrow{d^{\nabla}} \ml{D}^{\prime,0}(M,\ml{E})\xrightarrow{d^{\nabla}} \ml{D}^{\prime,1}(M,\ml{E})
\xrightarrow{d^{\nabla}}\ldots\xrightarrow{d^{\nabla}} \ml{D}^{\prime,n}(M,\ml{E})
\xrightarrow{d^{\nabla}} 0,$$
where we recall that $\ml{D}^{\prime,k}(M,\ml{E})$ is the topological dual of $\Omega^{n-k}(M,\ml{E})$. One can similarly define the $k$-th cohomology  
of that complex. A remarkable result of de Rham is that these two cohomologies 
coincide~\cite[Ch.~4]{dRh80}:
\begin{theo}[de Rham]\label{t:deRham} Let $u$ be an element in $\ml{D}^{\prime,k}(M,\ml{E})$ satisfying $d^{\nabla}u=0$. 
\begin{enumerate}
 \item There exists $\omega$ in $\Omega^k(M,\ml{E})$ such that $u-\omega$ belongs to $d^{\nabla}\left(\ml{D}^{\prime,k-1}(M,\ml{E})\right)$.
 \item If $u=d^{\nabla}v$ with $(u,v)$ in $\Omega^k(M,\ml{E})\times \mathcal{D}^{\prime,k-1}(M,\ml{E})$, then there exists 
 $\omega$ in $\Omega^{k-1}(M,\ml{E})$ such that $u=d^{\nabla}\omega$. 
\end{enumerate}
\end{theo}
\begin{rema}\label{r:regularity-derham}
We will in fact need something slightly more precise involving the anisotropic Sobolev spaces we have introduced -- see also~\cite[Lemma~2.1]{DyZw16} for related statements in terms of wavefront sets. 
More precisely, if $u$ belongs to $\ml{H}_k^{sm+n-k}(M,\ml{E})$, then $u-\omega$ belongs to $d^{\nabla}\ml{H}_{k-1}^{sm+n-k+1}(M,\ml{E})$. 
This observation will be crucial in our proof of Theorem~\ref{t:quasiisomorphism}.
\end{rema}

Let us for instance recall the argument from~\cite{DyZw16} which is based on elliptic regularity in order to verify that Remark~\ref{r:regularity-derham} indeed holds.

\begin{proof} First of all, fix a Riemannian metric $g$ on $M$ and a hermitian structure $\la.,.\ra_{\ml{E}}$ on $\ml{E}$. Following~\cite[Th.~4.11]{BDIP96}, we can compute 
the formal adjoint for the Hilbert structure on $L^2(M,\Lambda^k(T^*M)\otimes\ml{E})$. More precisely, for $\psi_1$ in $\Omega^{k}(M,\ml{E})$ and for 
$\psi_2$ in $\Omega^{k+1}(M,\ml{E})$, one has
$$\la d^{\nabla}\psi_1,\psi_2\ra_{L^2}=\int_{M}d^{\nabla}\psi_1\wedge \tau_{\ml{E}}\star_{k+1}(\psi_2)=(-1)^k\int_{M}\psi_1\wedge d^{\nabla^{\dagger}}\tau_{\ml{E}}\star_{k+1}(\psi_2),$$
 where the second equality follows from the Stokes Theorem. This tells us that\footnote{Under the additional assumptions that $\nabla$ preserves the Hermitian structure, we would have 
 $\tau_{\ml{E}}^{-1}d^{\nabla^{\dagger}}\tau_{\ml{E}}=d^{\nabla}$.} $(d^{\nabla})^*=(-1)^{k}\star_k^{-1}\tau_{\ml{E}}^{-1}d^{\nabla^{\dagger}}\tau_{\ml{E}}\star_{k+1}.$ We 
 then form the corresponding Laplace Beltrami operator $\Delta_{g,\ml{E}}=d^{\nabla}(d^{\nabla})^*+(d^{\nabla})^*d^{\nabla},$ which is formally selfadjoint. According to~\cite[p.19]{BDIP96}, the 
 principal symbol of this pseudodifferential operator is $-\|\xi\|_x^2\mathbf{Id}_{\Lambda^k(T^*M)\otimes\ml{E}}$. In particular, this defines an elliptic operator. We can now make use of 
 elliptic regularity to prove De Rham Theorem. 
 
 Fix $u$ in $\ml{D}^{\prime,k}(M,\ml{E})$ satisfying $d^{\nabla}u=0$. According to~\cite[Part.~I.3]{BDIP96}, there exists a 
 pseudodifferential operator $A_k$ of order $-2$ such that $u-\Delta_{g,\ml{E}}A_k u$ belongs to $\Omega^k(M,\ml{E})$. 
 Using the fact that $u$ is a cocycle, 
 we find that $d^{\nabla}\Delta_{g,\ml{E}}A_k u=\Delta_{g,\ml{E}}d^{\nabla}A_k u$ belongs to $\Omega^{k+1}(M)$. Thus, by elliptic regularity, $d^{\nabla}A_k u$ 
 belongs to $\Omega^{k+1}(M)$ and, by composing with $(d^{\nabla})^*$, $(d^{\nabla})^*d^{\nabla}A_k u$ belongs to $\Omega^{k}(M)$. Hence, we can find $\omega$ in $\Omega^k(M,\ml{E})$ such that 
 $u-\omega=d^{\nabla}(d^{\nabla})^*A_k u$ which proves point $(2)$. Note that, if $u$ belongs to $\ml{H}_k^{sm+n-k}(M,\ml{E})$, then we can verify from the composition rules of pseudodifferential 
 operators that $(d^{\nabla})^*A_k u$ belongs to $\ml{H}_{k-1}^{sm+n-k+1}(M,\ml{E})$ as pointed out in Remark~\ref{r:regularity-derham}.
 
 For the proof of point $(2)$, we proceed similarly except that we replace $u$ by $v$ which is not a priori a cocyle. Using the fact that 
 $u=d^{\nabla}v$ is smooth, we can still conclude that there exists $\omega\in \Omega^{k-1}(M,\ml{E})$ such that 
 $v-\omega=d^{\nabla}(d^{\nabla})^*A_k v$. Applying $d^{\nabla}$, we get the expected conclusion. 
\end{proof}

\subsection{Chain homotopy equation}\label{ss:chain-homotopy}
Thanks to~\eqref{e:exterior-derivative-sobolev} and to the fact that $d^{\nabla}\circ\ml{L}_{V,\nabla}^{(k)}=\ml{L}_{V,\nabla}^{(k+1)}\circ d^{\nabla}$, 
we can define, for every $z_0\in\IC$ the cohomological complex $(\text{Ker}(\ml{L}_{V,\nabla}^{(*)}+z_0)^{m_k(z_0)},d^{\nabla})$
$$0\xrightarrow{d^{\nabla}} \text{Ker}(\ml{L}_{V,\nabla}^{(0)}+z_0)^{m_0(0)}\xrightarrow{d^{\nabla}} 
\text{Ker}(\ml{L}_{V,\nabla}^{(1)}+z_0)^{m_1(0)}\xrightarrow{d^{\nabla}} \ldots
\xrightarrow{d^{\nabla}} \text{Ker}(\ml{L}_{V,\nabla}^{(n)}+z_0)^{m_n(0)}\xrightarrow{d^{\nabla}} 0.$$
By construction, this cohomological complex coincide with the complex $(C_{V,\nabla}^{\bullet}(z_0),d^{\nabla})$ 
defined by~\eqref{e:ruelle-complex}. Observe that, for $z_0\neq 0$, $(C_{V,\nabla}^{\bullet}(z_0),d^{\nabla})$ is 
\textbf{acyclic} and we shall come back to this case later on. Before that, we focus on the case $z_0=0$ and we 
start with a few preliminary observations. First, as $(\ml{L}_{V,\nabla}+z)$ commutes with $d^{\nabla}$ and $\iota_V$, 
we can verify from~\eqref{e:spectral-proj} that, for every $0\leq k\leq n$,
\begin{equation}\label{e:d-commute-projector}
 d^{\nabla}\circ\pi_0^{(k)}=\pi_0^{(k+1)}\circ d^{\nabla}\ \text{and}\ \iota_V\circ\pi_0^{(k+1)}=\pi_0^{(k)}\circ \iota_V.
\end{equation}
Hence, one has
\begin{equation}\label{e:Lie-commute-projector}
 \ml{L}_{V,\nabla}^{(k)}\circ\pi_0^{(k)}=\pi_0^{(k)}\circ \ml{L}_{V,\nabla}^{(k)}.
\end{equation}
According to~\cite[Appendix~A.18]{HeSj86}, the operator $\ml{L}_{V,\nabla}^{(k)}$ is invertible on the space 
$$\ml{H}_k^{sm+n-k}(M,\ml{E})\cap\text{Ker}\pi_0^{(k)}=(\text{Id}-\pi_0^{(k)})\ml{H}_k^{sm+n-k}(M,\ml{E}).$$
The same holds in degree $k+1$. Hence, thanks to the relation
$$d^{\nabla}\circ\ml{L}_{V,\nabla}^{(k)}=\ml{L}_{V,\nabla}^{(k+1)}\circ d^{\nabla}:\ml{H}_k^{sm+n-k}(M,\ml{E})\cap\text{Ker}\pi_0^{(k)}
\rightarrow\ml{H}_{k+1}^{sm+n-(k+1)}(M,\ml{E})\cap\text{Ker}\pi_0^{(k+1)},$$
we can deduce that
\begin{equation}\label{e:commute-inverse}
 (\ml{L}_{V,\nabla}^{(k+1)})^{-1} \circ d^{\nabla}\circ(\text{Id}_{\ml{H}^{sm+n-k}_k(M,\ml{E})}-\pi_0^{(k)})
 =d^{\nabla}\circ(\ml{L}_{V,\nabla}^{(k)})^{-1}\circ (\text{Id}_{\ml{H}^{sm+n-k}_k(M,\ml{E})}-\pi_0^{(k)}).
\end{equation}
Now that we have settled these equalities, we write, for any $\psi_1$ 
in $\ml{H}_k^{sm+n-k}(M,\ml{E})$,
\begin{eqnarray*}
 \psi_1 &= &\pi_0^{(k)}(\psi_1)+(\text{Id}_{\ml{H}^{sm+n-k}_k(M,\ml{E})}-\pi_0^{(k)})(\psi_1)\\
  & = &\pi_0^{(k)}(\psi_1)+(d^{\nabla}\circ \iota_V+\iota_V\circ d^{\nabla})\circ (\ml{L}_{V,\nabla}^{(k)})^{-1}\circ(\text{Id}_{\ml{H}^{sm+n-k}_k(M,\ml{E})}-
  \pi_0^{(k)})(\psi_1)\\
& = &\pi_0^{(k)}(\psi_1)+d^{\nabla}\circ \left(\iota_V\circ (\ml{L}_{V,\nabla}^{(k)})^{-1}\circ(\text{Id}_{\ml{H}^{sm+n-k}_k(M,\ml{E})}-\pi_0^{(k)})\right)(\psi_1)\\
 & &+\left(\iota_V\circ (\ml{L}_{V,\nabla}^{(k)})^{-1}\circ(\text{Id}_{\ml{H}^{sm+n-(k+1)}_{k+1}(M,\ml{E})}-\pi_0^{(k+1)})\right)\circ d^{\nabla}(\psi_1),
\end{eqnarray*}
where the last equality follows from~\eqref{e:d-commute-projector} and~\eqref{e:commute-inverse}. Thus, if we set
\begin{equation}\label{e:chain-contraction-map}R^{(k)}:=\iota_V^{(k)}\circ (\ml{L}_V^{(k)})^{-1}\circ(\text{Id}_{\ml{H}^{sm+n-k}_k(M,\ml{E})}-\pi_0^{(k)})
:\ml{H}^{sm+n-k}_k(M,\ml{E})\rightarrow \ml{H}^{sm+n-k}_{k-1}(M,\ml{E}),
\end{equation}
we obtain the following \textbf{chain homotopy equation}:
\begin{equation}\label{e:chain-homotopy}\forall\psi_1\in\ml{H}^{sm+n-k}_k(M,\ml{E}),\ \psi_1=\pi_0^{(k)}(\psi_1)
 +d^{\nabla}\circ R^{(k)}(\psi_1)+R^{(k+1)}\circ d^{\nabla}(\psi_1).
\end{equation}
\begin{rema}\label{r:chain-contraction-map} Note that the map $R$ will play a crucial in our calculations related to the torsion. 
By a similar argument as the one used to prove~\eqref{e:commute-inverse}, we can verify that
$$(\ml{L}_{V,\nabla}^{(k)})^{-1} \circ \iota_V\circ(\text{Id}_{\ml{H}^{sm+n-k}_k(M,\ml{E})}-\pi_0^{(k)})
 =\iota_V\circ(\ml{L}_{V,\nabla}^{(k)})^{-1}\circ (\text{Id}_{\ml{H}^{sm+n-k}_k(M,\ml{E})}-\pi_0^{(k)}).$$
In particular, for every $z_0\neq 0$, $R$ is a 
\textbf{chain contraction map} for the acyclic complex $(C_{V,\nabla}^{\bullet}(z_0),d^{\nabla})$~\cite{Mn14}, in the sense that 
$\text{Id}=(d^{\nabla}+R)^2$ and $R^2=0$ for that complex.
\end{rema}

\subsection{Proof of Theorem~\ref{t:quasiisomorphism}}

We can now give the proof of Theorem~\ref{t:quasiisomorphism}. Thanks to~\eqref{e:d-commute-projector}, the induced maps 
on the cohomology are well defined. Suppose now that $\psi$ is a cocycle in 
$\Omega^k(M,\ml{E})$ which verifies $\pi_0^{(k)}(\psi)=0$. Thanks to~\eqref{e:chain-homotopy}, we deduce that
$$\psi=d^{\nabla}\circ R^{(k)}(\psi).$$
From the second part of de Rham's Theorem, we deduce that there exists $\omega\in\Omega^{k-1}(M,\ml{E})$ such that $\psi=d\omega$. Hence, $\psi$ 
is a coboundary. This shows the injectivity of the map. Consider now surjectivity and fix $\psi$ a cocycle in 
$\operatorname{Ker}(\ml{L}_V^{(k)})^{m_k(0)}$. From 
the first of de Rham's Theorem (and more precisely Remark~\ref{r:regularity-derham}), 
there exists a cocycle $\omega$ in $\Omega^k(M,\ml{E})$ such that $\psi-\omega$ belongs to 
$d^{\nabla}\ml{H}_{k-1}^{sm+n-(k-1)}(M,\ml{E})$. Using~\eqref{e:chain-homotopy} one more time, we can write
$$\omega=\pi_0^{(k)}(\omega)+d^{\nabla}\circ R^{(k)}(\omega).$$
Hence, we deduce that $\psi-\pi_0^{(k)}(\omega)$ is in the range of $\ml{H}^{sm+n-k}_{k-1}(M,\ml{E})$ under $d^{\nabla}$. 
Using the second part of de Rham's Theorem, there exists $\tilde{\omega}$ in $\ml{H}^{sm+n-k}_{k-1}(M,\ml{E})$ such that
$$\psi=\pi_0^{(k)}(\omega)+d^{\nabla}\tilde{\omega}.$$
Recall that the Pollicott-Ruelle spectrum is independent of the choice of the anisotropic Sobolev 
space~\cite[Th.~1.5]{FaSj11}. In particular, the resonant state associated with $z_0=0$ are the same for the 
spaces $\ml{H}_{\bullet}^{sm+n-\bullet}$ and $\ml{H}_{\bullet}^{sm+n-\bullet-1}$. Hence, applying the spectral projector $\pi_0^{(k)}$ 
and implementing relation~\eqref{e:d-commute-projector}, we find that
$$\psi=\pi_0^{(k)}(\omega)+d^{\nabla}\circ\pi_0^{(k-1)}(\tilde{\omega}),$$
which concludes the proof of the surjectivity.

\subsection{Morse type inequalities}\label{ss:morse} In this paragraph, we will explain how to prove Corollary~\ref{c:morse-smale-inequality}. The 
proof of Corollary~\ref{c:morse-spectral} follows from a standard argument on finite dimensional cohomological complexes and we omit 
it -- see~\cite[Sect.~8]{DaRi16} for a brief reminder. In the case of Corollary~\ref{c:morse-smale-inequality}, we remind that, in the case where $V$ 
is a Morse-Smale vector field which is $\ml{C}^{\infty}$-diagonalizable, the main result of~\cite{DaRi17b} (see paragraphs~\ref{sss:local-coordinates-critical-point} 
and~\ref{sss:local-coordinates-closed-orbit} for a brief reminder) implies that 
the dimension of $C^k_{\nabla,V}(0)$ is equal to 
$$Nc_k(V)+\sum_{\Lambda:\text{dim} W^s(\Lambda)=k}m_{\Lambda}+\sum_{\Lambda:\text{dim} W^s(\Lambda)=k-1}m_{\Lambda}.$$
Thus, under the additional assumption that $V$ is $\ml{C}^{\infty}$-diagonalizable, we can conclude the proof of the Corollary~\ref{c:morse-smale-inequality} thanks to Corollary~\ref{c:morse-spectral}. 
Fix now a general 
Morse-Smale vector field $V_0$ on $M$. Knowing that the set of Morse-Smale vector fields is open inside the set of smooth vector fields, we can deduce that, 
for every $V$ in a small neighborhood of $V_0$, the vector field is still Morse-Smale. Thanks to the Sternberg-Chen Theorem~\cite{Ne69, WWL08} 
(see also~\cite[Appendix]{DaRi17a} for a brief reminder), we know that the assumption of being $\ml{C}^{\infty}$-diagonalizable is satisfied as soon as all 
the eigenvalues of the linearized systems satisfies certain nonresonance assumptions and as soon as they are distinct. In particular, all these conditions are satisfied by a 
dense subset of vector fields and we can find $V$ arbitrarily close to $V_0$ which is Morse-Smale and $\ml{C}^{\infty}$-diagonalizable. In order to conclude, we 
observe that, for a small enough perturbation, we do not modify the number of critical elements of the flow of index $k$ and every 
closed orbit stays in the same homotopy class. In particular, all the terms in the sum defining the upper bound of corollary~\ref{c:morse-smale-inequality} are equal for $V$ and $V_0$.

\subsection{Trace asymptotics and spectral determinants}\label{ss:trace} Let us now draw some nice spectral interpretation of these results. Fix $U$ a bounded open set in $\IC$. We define 
$$\Pi_U^{(k)}=\sum_{z_0\in U}\pi_{z_0}^{(k)}.$$
Note that, as the spectrum is discrete, this defines a finite sum. We now write
$$C_{V,\nabla}^{\text{even}}(z_0)=\bigoplus_{k=0\ \text{mod}\ 2}C_{V,\nabla}^{k}(z_0)\quad \text{and}\quad C_{V,\nabla}^{\text{odd}}(z_0)=\bigoplus_{k=1\ \text{mod}\ 2}C_{V,\nabla}^{k}(z_0).$$
Thanks to Remark~\ref{r:chain-contraction-map}, we can note that $Q=d^{\nabla}+R$ defines an operator which exchanges the chiralities and that, for $z_0\neq 0$, $Q$ is an isomorphism 
on $C_{V,\nabla}^{\text{even}}(z_0)\oplus C_{V,\nabla}^{\text{odd}}(z_0)$ (as $Q^2=\text{Id}$). In particular, 
$$\forall z_0\neq 0,\ \text{dim}\ C_{V,\nabla}^{\text{even}}(z_0)=\text{dim}\ C_{V,\nabla}^{\text{odd}}(z_0).$$
Aguing as in~\cite[Sect.~7]{DaRi16}, we can then derive from the case of equality in Corollary~\ref{c:morse-spectral}:
\begin{coro}\label{c:spectral-det} Suppose the assumptions of Theorem~\ref{t:quasiisomorphism} are satisfied. Then, for every bounded open set $U$ in $\IC$ containing $0$ and for every $z\in\IC^*$,
 $$\prod_{k=0}^n\operatorname{det}\left(\Pi_U^{(k)}\left(z+\ml{L}_{V,\nabla}^{(k)}\right)\Pi_U^{(k)}\right)^{(-1)^k}=z^{\chi(M,\ml{E})},$$
 and, for all $t>0$,
 $$\sum_{k=0}^n(-1)^k\operatorname{Tr}\left(\Pi_U^{(k)}e^{-t\ml{L}_{V,\nabla}^{(k)}}\Pi_U^{(k)}\right)=\chi(M,\ml{E}).$$
\end{coro}

\subsection{Poincar\'e duality}\label{s:poincare}

Consider now the dual complex associated with the vector field $-V$ and the connection $\nabla^{\dagger}$, i.e.
$$0\xrightarrow{d^{\nabla^{\dagger}}} \text{Ker}(\ml{L}_{-V,\nabla^{\dagger}}^{(0)})^{m_d(0)}\xrightarrow{d^{\nabla^{\dagger}}} \text{Ker}(\ml{L}_{-V,\nabla^{\dagger}}^{(1)})^{m_{n-1}(0)}
\xrightarrow{d^{\nabla^{\dagger}}} \ldots
\xrightarrow{d^{\nabla^{\dagger}}} \text{Ker}(\ml{L}_{-V,\nabla^{\dagger}}^{(n)})^{m_0(0)}\xrightarrow{d^{\nabla^{\dagger}}} 0.$$
From~\eqref{e:dual} and from section~\ref{s:spectral}, we know that the operator $-\ml{L}_{-V,\nabla^{\dagger}}^{(k)}$ can be identified with the 
dual of $-\ml{L}_{V,\nabla}^{(n-k)}.$ These two complexes are dual to each other via the duality bracket between $\ml{H}^{m+k}_k(M,\ml{E})$ and 
$\ml{H}_{n-k}^{-(m+k)}(M,\ml{E}')$. In other words, we have, for every $\psi_1$ in $\text{Ker}(\ml{L}_{V,\nabla}^{(k)})^{m_k(0)}$ and 
every $\psi_2$ in $\text{Ker}(\ml{L}_{-V,\nabla^{\dagger}}^{(n-k)})^{m_k(0)}$,
$$\la \psi_1,\psi_2\ra=\la \psi_1,\psi_2\ra_{\ml{H}_{k}^{m+k}(M,\ml{E})\times\ml{H}_{n-k}^{-(m+k)}(M,\ml{E}')}=\int_M\psi_2\wedge\psi_1.$$
We can then define a Poincar\'e map:
$$\ml{P}_0^{(k)}:\psi\in\text{Ker}(\ml{L}_{V,\nabla}^{(k)})^{m_k(0)}\mapsto\la\psi,.\ra\in\left(\text{Ker}(\ml{L}_{-V,\nabla^{\dagger}}^{(n-k)})^{m_k(0)}\right)'.$$
By a classical argument for complex which are dual to each other -- see~\cite[Paragraph~8.5]{DaRi16} for a brief reminder, we can then deduce that 
$\ml{P}_0^{(k)}$ induces an isomorphism between the $k$-th cohomological group of 
$(\operatorname{Ker}(\ml{L}_{V,\nabla}^{(*)})^{m_*(0)},d^{\nabla})$ and the dual of 
the $(n-k)$-th cohomological group of $(\operatorname{Ker}(\ml{L}_{-V,\nabla^{\dagger}}^{(*)})^{m_*(0)},d^{\nabla^{\dagger}})$. Combining this 
with Theorem~\ref{t:quasiisomorphism}, one can recover the following version of Poincar\'e duality:
$$\forall 0\leq k\leq n,\ b_k(M,\ml{E})=b_{n-k}(M,\ml{E}').$$

\section{Koszul homological complex}\label{s:koszul}

As a warm-up for the calculation of the spectral torsion, we start with some consideration on the following homological complex:
\begin{equation}\label{e:ruelle-koszul-complex}0\xrightarrow{\iota_V} C_{V,\nabla}^n(0)\xrightarrow{\iota_V} C_{V,\nabla}^{n-1}(0)\xrightarrow{\iota_V} \ldots\xrightarrow{\iota_V} 
C_{V,\nabla}^0(0)\xrightarrow{\iota_V} 0.\end{equation}
This complex will be refered as the \textbf{Morse-Smale-Koszul} complex and we will verify that, in the case of Morse-Smale vector fields, its homology counts the number of critical points in each degree. More precisely, 
we will prove:
\begin{prop}\label{p:koszul} Suppose that $V$ is a Morse-Smale vector field which is $\ml{C}^{\infty}$-diagonalizable and that $\nabla$ 
preserves a hermitian 
structure on $\ml{E}$ (which is of dimension $N$). 

Then, the dimension of the homology of degree $k$ of $(C^{\bullet}_{V,\nabla},\iota_V)$ is equal to 
$Nc_k(V)$, where $c_k(V)$ is the number of critical points such that $\operatorname{dim} W^s(\Lambda)=k$. 

\end{prop}
This proposition proves\footnote{This follows from 
Corollary~\ref{c:morse-smale-inequality}.} Theorem~\ref{t:koszul} and it generalizes the results of~\cite[Sect.~8.6]{DaRi16} which were only valid for gradient flows and for 
the trivial bundle $M\times\IC$. We refer 
to~\cite{DaRi16} for more details and examples on Koszul complexes. Recall in particular that, for any smooth vector field $V$ with nondegenerate 
critical points, the homology of degree $0$ of $(\Omega^{\bullet}(M,\ml{E}),\iota_V)$ is equal to the total number of critical points while the 
homology in higher degree vanishes. Here, the Morse-Smale-Koszul complex allows to compute something slightly more precise as its homology 
allows to count the critical points of any degree.

In all this section, we will suppose that $V$ is a Morse-Smale vector field which is $\ml{C}^{\infty}$-diagonalizable and that $\nabla$ preserves 
a Hermitian structure on $\ml{E}$.

\subsection{Spectrum of Morse-Smale flows on the imaginary axis}

Let us first use our assumptions that the flow is smoothly diagonalizable in order to put it into a normal form near the critical 
elements $\Lambda$ of the vector field $V$, namely its closed orbits and critical points. Thanks to our results 
from~\cite[Sect.~7]{DaRi17b}, we can write down the local expression of the Pollicott-Ruelle resonant 
states in this system of normal coordinates.

\subsubsection{Critical points}\label{sss:local-coordinates-critical-point} Suppose that $\Lambda$ is a critical point of 
$V$ of index $k$, i.e. such that $\text{dim}\ W^s(\Lambda)=k$. Thanks to the fact that the flow is supposed to be smoothly 
diagonalizable, we can find a system of smooth coordinates $(x,y)\in\IR^k\times\IR^{n-k}$ near $\Lambda$ such that the flow $\varphi^t$ induced by $V$ 
can be written as follows:
$$\varphi^t(x,y)=\left(e^{tA_s^{\Lambda}}x,e^{tA_u^{\Lambda}}y\right),$$
where $A_s^{\Lambda}$ (resp. $A_u^{\Lambda}$) belongs to $M_k(\IR)$ (resp. $M_{n-k}(\IR)$), is diagonalizable (in $\IC$) and has all its eigenvalues on the half 
plane $\text{Re} (z)<0$ (resp. $\text{Re}(z)>0$). In particular, the local stable manifold of $\Lambda$ corresponds to the set $y=0$ while the 
unstable manifold is given by $x=0$. Near $\Lambda$, we can also construct a moving frame $(\mathbf{c}_1^{\Lambda},\ldots,\mathbf{c}_N^{\Lambda})$ 
of $\ml{E}$ such that $\nabla c_{j}^{\Lambda}=0$ for every $0\leq j\leq N$ -- see e.g.~\cite[Th.~12.25]{Lee09}. According 
to~\cite[Sect.~7]{DaRi17b}, one can associate to every such critical element of $V$ a family of $N$ Pollicott-Ruelle resonant states of degree $k$ associated with 
the resonance $z_0=0$. We denote them by
$$U_{1}^{\Lambda},\ldots, U_N^{\Lambda},$$
and their local form near $\Lambda$ is given, for every $1\leq j\leq N$ by
$$U_j^{\Lambda}(x,y,dx,dy)=\delta_0(x_1,\ldots, x_k)dx_1\wedge\ldots\wedge dx_k \otimes \mathbf{c}_j^{\Lambda}(x,y).$$
Moreover, one knows that the support of $U_j^{\Lambda}$ is equal to $\overline{W^u(\Lambda)}$

\subsubsection{Closed orbits}\label{sss:local-coordinates-closed-orbit} Suppose that $\Lambda$ is a closed orbit of 
$V$ of index $k$, i.e. such that $\text{dim}\ W^s(\Lambda)=k$. The situation is now slightly more complicated. We denote by $\ml{P}_{\Lambda}$ 
the minimal period of the closed orbit. Thanks to the hypothesis that $V$ is $\ml{C}^{\infty}$-diagonalizable, we can find a system of 
smooth coordinates $(x,y,\theta)\in\IR^{k-1}\times\IR^{n-k}\times(\IR/\ml{P}_{\Lambda}\IZ)$ near $\Lambda$ such that the flow can be put 
under the following normal form
$$\varphi^t(x,y)=\left(P(\theta+t)e^{tA}P(\theta)^{-1}(x,y),\theta+t\right),$$
where
\begin{itemize}
 \item $A$ is of the form $\text{diag} (A_s^{\Lambda},A_u^{\Lambda})$ with $A_s^{\Lambda}$ 
 (resp. $A_u^{\Lambda}$) belonging to $M_{k-1}(\IR)$ (resp. $M_{n-k}(\IR)$) and diagonalizable (in $\IC$). Moreover, $A_s^{\Lambda}$ (resp. $A_u^{\Lambda}$) 
 has all its eigenvalues on the half plane $\text{Re} (z)<0$ (resp. $\text{Re}(z)>0$).
 \item $P(\theta)$ is $2\ml{P}_{\Lambda}$-periodic. Moreover, it satisfies $P(0)=\text{Id}_{\IR^{n-1}}$, $J_{\Lambda}:=P(\ml{P}_{\Lambda})=\text{diag}(\pm 1)$ and 
 $P(\theta+\ml{P}_{\Lambda})=J_{\Lambda}P(\theta)$.
\end{itemize}
The manifold $W^s(\Lambda)$ (and thus $W^u(\Lambda)$) is orientable if and only if the determinant of $J_{\Lambda}$ restricted to $\IR^{k-1}\oplus\{0\}$ 
is equal to $1$. We set $\varepsilon_{\Lambda}=0$ whenever $\Lambda$ is orientable and $\varepsilon_{\Lambda}=1/2$ otherwise. It is related to the twisting 
index by the identity $\Delta_{\Lambda}=e^{2i\pi\varepsilon_{\Lambda}}$. 
In this system of coordinates, $\Lambda$ is given by the set $\{(0,0),\theta)\}$, the stable manifold by $\{P(\theta)(x,0),\theta)\}$ 
and the unstable one by $\{P(\theta)(0,y),\theta)\}$. As for critical points, we can construct a moving frame 
$(\mathbf{c}_1^{\Lambda},\ldots,\mathbf{c}_N^{\Lambda})$ of $\ml{E}$ such that there exists $\gamma_1^{\Lambda},\ldots,\gamma_N^{\Lambda}$ in $[0,1)$ for which one has
$\nabla \mathbf{c}_{j}^{\Lambda}=\frac{2i\pi\gamma_j^{\Lambda}}{\ml{P}_{\Lambda}}\mathbf{c}_{j}^{\Lambda} d\theta$ 
for every $0\leq j\leq N$ -- see e.g.~\cite[Sect.~4]{DaRi17b} for details. Here, we note that 
$\left(e^{2i\pi\gamma_j^{\Lambda}}\right)_{j=1}^N$ are the 
eigenvalues of the monodromy matrix $M_{\ml{E}}(\Lambda)$ for the parallel transport around $\Lambda$.

Following one more time our previous work~\cite[Sect.~7]{DaRi17b}, we can associate families of Pollicott-Ruelle resonant states of degree $k-1$ and $k$. They are now 
indexed by $(p,j)\in\IZ\times\{1,\ldots, N\}$ and we denote them by $U_{j,p}^{\Lambda}$ and $\tilde{U}_{j,p}^{\Lambda}$. Their local form near $\Lambda$ 
is given by
$$U_{j,p}^{\Lambda}=e^{\frac{2i\pi( p+\varepsilon_{\Lambda})\theta}{\ml{P}_{\Lambda}}}
\left(P(\theta)^{-1}\right)^*\left(\delta_0(x_1,\ldots, x_{k-1})dx_1\wedge\ldots\wedge dx_{k-1}\right) \otimes \mathbf{c}_j^{\Lambda}(x,y,\theta),$$
and
$$\tilde{U}_{j,p}^{\Lambda}=e^{\frac{2i\pi (p+\varepsilon_{\Lambda})\theta}{\ml{P}_{\Lambda}}}
\left(P(\theta)^{-1}\right)^*\left(\delta_0(x_1,\ldots, x_{k-1})dx_1\wedge\ldots\wedge dx_{k-1}\right)\wedge d\theta\otimes \mathbf{c}_j^{\Lambda}(x,y,\theta).$$
For a fixed $(j,p)$, the resonant states $U_{j,p}^{\Lambda}$ and $\tilde{U}_{j,p}^{\Lambda}$ are associated with the resonance 
$$z_0=-\frac{2i\pi}{\ml{P}_{\Lambda}}\left(p+\varepsilon_{\Lambda}+\gamma_j^{\Lambda}\right),$$
and the support of $U_{j,p}^{\Lambda}$ is equal to $\overline{W^u(\Lambda)}$

\subsubsection{Resonant states on the imaginary axis}

Finally, the main result of~\cite[Sect.~7]{DaRi17b} states that the sets
$$\bigcup_{\Lambda\ \text{critical point}:\text{dim} W^s(\Lambda)=k}\left\{U_{j}^{\Lambda}:1\leq j\leq N\right\},$$
$$\bigcup_{\Lambda\ \text{closed orbit}:\text{dim} W^s(\Lambda)=k}\left\{\tilde{U}_{j,p}^{\Lambda}:1\leq j\leq N,\ p\in\IZ\right\},$$
and
$$\bigcup_{\Lambda\ \text{closed orbit}:\text{dim} W^s(\Lambda)=k+1}\left\{U_{j,n}^{\Lambda}:1\leq j\leq N,\ p\in\IZ\right\}$$
generate all the Pollicott-Ruelle resonances of degree $k$ associated with a resonance lying on the imaginary axis. Moreover, all 
these states are linearly independent. We also recall from~\cite{DaRi17b} that each of these states verifies 
$(\ml{L}_{V,\nabla}-z_0)U^{\Lambda}=0$ on $W^u(\Lambda)$ but that $U^\Lambda$ is only a generalized eigenvalue equation 
$(\ml{L}_{V,\nabla}-z_0)^mU^{\Lambda}=0$ on $M$.

\begin{rema}\label{r:trivial-kernel} 
We note that there is no resonant state of degree $k$ associated with $z_0=0$ if the Morse-Smale flow is nonsingular and if
$\Delta_{\Lambda}$ is not an eigenvalue of $M_{\ml{E}}(\Lambda)$ for every closed orbit $\Lambda$. 
\end{rema}

\subsection{Proof of Proposition~\ref{p:koszul}}\label{ss:proof-koszul} In order to prove Proposition~\ref{p:koszul}, we start by fixing some critical point $\Lambda$ of index $k$ and 
$1\leq j\leq N$. Applying the contraction operator $\iota_V$, we find that $\iota_V(U_j^{\Lambda})$ is equal to $0$ near $\Lambda$. Thus, as 
$\iota_V(U_j^{\Lambda})$ verifies $\ml{L}_{V,\nabla}\circ\iota_V(U_{j}^{\Lambda})=0$ on $W^u(\Lambda)$ and as it is supported on $\overline{W^u(\Lambda)}$, we find 
that $\iota_V(U_j^{\Lambda})$ is supported on $\overline{W^u(\Lambda)}-W^u(\Lambda)$. As $(C^{\bullet}_{V,\nabla},\iota_V)$ is an homological complex,
we know that $\iota_V(U_j^{\Lambda})$ will be a linear combination of elements inside $C_{V,\nabla}^{k-1}(0)$. According to the above 
discussion, such elements are supported either by the closure of unstable manifolds of dimension $n-k+1$ (in the case of a critical point) or
by the closure of unstable manifolds of dimension $n-k+1$ or $n-k+2$ (in the case of closed orbits). 
Yet, according to Smale~\cite[Lemma~3.1]{Sm60} -- see~\cite{DaRi17a} for a brief reminder, such unstable submanifolds cannot be contained in the closure of 
$W^u(\Lambda)$. It means that the linear combination is equal to $0$, i.e. $\iota_V(U_j^{\Lambda})=0$.

Fix now a closed orbit $\Lambda$ such that\footnote{Hence, $\text{dim}\ W^u(\Lambda)=n-k+1$.} $\text{dim}\ W^s(\Lambda)=k$ and $(j,p)$ such that $U_{j,p}^{\Lambda}$ and $\tilde{U}_{j,p}^{\Lambda}$ are associated with the resonance $0$. Let 
us first compute $\iota_V(U_{j,p}^{\Lambda})$. 
As before, we find, by using the local form of 
the resonant state and the eigenvalue equation on 
$W^u(\Lambda)$, that $\iota_V(U_{j,p}^{\Lambda})$ 
is supported on $\overline{W^u(\Lambda)}-W^u(\Lambda)$ which is  
the union of unstable manifolds whose dimension is $\leq n-k+1$
thanks to Smale's 
Lemma~\cite[Lemma~3.1]{Sm60}. 
Recall also that if $W^u(\Lambda')\subset\overline{W^u(\Lambda)}-W^u(\Lambda)$ with 
$\text{dim}\ W^u(\Lambda')=n-k+1$, then $\Lambda'$ is a closed orbit. 
As we have a homological complex, $\iota_V(U_{j,p}^{\Lambda})$ is an element in 
$C^{k-2}_{V,\nabla}(0)$. From the description of the resonant state we just gave, we can verify as above that 
this subspace is generated by critical elements associated with unstable manifolds of dimension $\geq n-k+2$. This implies that 
$\iota_V(U_{j,p}^{\Lambda})=0$. It remains to compute $\iota_V(\tilde{U}_{j,p}^{\Lambda})$ which is now an element in 
$C^{k-1}_{V,\nabla}(0)$. Contrary to the previous cases, $\iota_V(\tilde{U}_{j,p}^{\Lambda})$ is supported by $\overline{W^u(\Lambda)}$ and it 
is equal to $U_{j,p}^{\Lambda}$ near $\Lambda$. Hence, from the eigenvalue equation, $\iota_V(\tilde{U}_{j,p}^{\Lambda})-U_{j,p}^{\Lambda}$ is now supported 
by $\overline{W^u(\Lambda)}-W^u(\Lambda)$ which is a union of unstable manifold of dimension $\leq n-k+1$. Thus, using Smale's result one more time, we can conclude by similar arguments that 
$$\iota_V(\tilde{U}_{j,p}^{\Lambda})=U_{j,p}^{\Lambda}+\sum_{(\Lambda',j',p'):\Lambda'\neq\Lambda,\ \text{dim} W^s(\Lambda')=k}
\alpha_{\Lambda',j',p'}U_{j',p'}^{\Lambda'}.$$ 

In any degree, we can finally conclude that $$\text{Ker}\left(\iota_V^{(k)}\right)/\text{Ran}\left(\iota_V^{(k+1)}\right)\simeq
\text{span}\left\{U_j^{\Lambda}:\ \text{dim}\Lambda=k,\ 1\leq j\leq N\right\},$$
where the kernel and the range are understood as for operators acting on $C^{\bullet}_{V,\nabla}(0)$.

\subsection{The case $z_0\neq 0$}\label{ss:koszul-nonzero}

As for the case of the coboundary operator, we can verify that the homology of the complex $(C^{\bullet}_{V,\nabla}(z_0),\iota_V)$ is trivial for 
any $z_0\neq 0$. If we restrict to the case where $z_0\in i\IR^*$, we can in fact be slightly more precise. In fact, we should emphasize that, 
in the second part of the argument, we did not use the fact that $z_0$ is equal to $0$ and the argument is in fact valid for any 
 resonance $z_0$ lying on the imaginary axis $i\IR^*$. More precisely, for every $(\Lambda,j,p)$ such that
 $$z_0=-\frac{2i\pi(p+\varepsilon_{\Lambda}+\gamma_j^{\Lambda})}{\ml{P}_{\Lambda}},$$
we will have $\iota_V(U_{j,p}^{\Lambda})=0$ and
$$\iota_V(\tilde{U}_{j,p}^{\Lambda})=U_{j,p}^{\Lambda}+\sum_{(\Lambda',j',p'):\Lambda'\neq\Lambda,\ \text{dim} W^s(\Lambda')=k}
\alpha_{\Lambda',j',p'}U_{j',p'}^{\Lambda'}.$$ 
In particular, we have, for $z_0\in i\IR^*$ and for all $0\leq k\leq n$,
$$\text{Ker}\left(\iota_V^{(k)}\rceil_{C^k_{V,\nabla}(z_0)}\right)=\text{span}\left\{U_{j,p}^{\Lambda}:\ \text{dim} W^s(\Lambda)=k+1\ 
\text{and}\ z_0=-\frac{2i\pi}{\ml{P}_{\Lambda}}\left(p+\varepsilon_{\Lambda}+\gamma_j^{\Lambda}\right)\right\}.$$

\section{Torsion of Pollicott-Ruelle resonant states}\label{s:torsion}

In this section, we will discuss the topological properties of the resonant states lying on the imaginary axis together with the proof of Theorem~\ref{t:fried0} 
which is related to the spectral zeta functions associated with the Pollicott-Ruelle resonances on the imaginary axis.


In the present section, we first present 
a new definition of
torsion for the infinite dimensional chain complex
of Pollicott--Ruelle resonances states corresponding
by the spectrum on the imaginary axis which is directly related
to Theorem~\ref{t:fried0}. 
 More specifically, we know from paragraph~\ref{ss:chain-homotopy} that, 
for every $z_0\in i\IR^*$,
$$0\xrightarrow{d^{\nabla}} C_{V,\nabla}^{0}(z_0)\xrightarrow{d^{\nabla}} 
C_{V,\nabla}^{1}(z_0)\xrightarrow{d^{\nabla}} \ldots
\xrightarrow{d^{\nabla}} C_{V,\nabla}^{n}(z_0)\xrightarrow{d^{\nabla}} 0.$$
defines an acyclic cohomological complex. Moreover, we constructed in~\cite{DaRi17b} a natural basis 
for this complex -- see 
paragraph~\ref{sss:local-coordinates-closed-orbit} for a brief reminder. Recall now that, as soon
as we are given a cohomological 
complex with a preferred basis, we can compute its torsion which is just a certain determinant associated with $d^{\nabla}$ and computed 
in this basis. Hence, we will first 
compute the torsion of each individual complex $(C_{V,\nabla}^\bullet(z_0),d^{\nabla})$ for a purely imaginary resonance $z_0$ 
in the basis of 
paragraph~\ref{sss:local-coordinates-closed-orbit}. 
Then, mimicking standard procedures used in spectral theory and QFT, we will define 
the \textbf{torsion of the infinite dimensional}
complex $\bigoplus_{z_0\in i\mathbb{R}} (C_{V,\nabla}^\bullet(z_0),d^{\nabla}) $ as the
zeta regularized infinite product of the individual torsions in order 
to compute the torsion of the complex carried by the 
entire imaginary axis. Finally, we will prove Theorem~\ref{t:fried0} which is in fact related to the above regularized torsion.

From this point on, we will suppose that $V$ is a Morse-Smale vector field which is $\ml{C}^{\infty}$-diagonalizable 
and that $\nabla$ preserves a Hermitian structure on $\ml{E}$.

\subsection{Torsion of individual complexes $(C_{V,\nabla}^{\bullet}(z_0),d^{\nabla})$} As was already pointed out, the finite dimensional 
complex $(C_{V,\nabla}^{\bullet}(z_0),d^{\nabla})$ is acyclic as soon as $z_0\neq 0$. Recall also from 
paragraph~\ref{sss:local-coordinates-closed-orbit} that it has a preferred basis which is associated to the unstable 
manifolds of the closed orbits of the vector field $V$. Thus, given $z_0\neq 0$ and its preferred basis of 
Pollicott-Ruelle resonant states, it is natural to compute the torsion of the complex 
$(C_{V,\nabla}^{\bullet}(z_0),d^{\nabla})$~\cite{F87, Mn14} and we shall start with this preliminary calculation when $z_0\in i\IR^*$.

Using the conventions of paragraph~\ref{sss:local-coordinates-closed-orbit}, we have, for every $0\leq k\leq n$,
$$\bigoplus_{k=0}^nC^k_{V,\nabla}(z_0)=\text{span}\left\{U_{j,p}^{\Lambda},\tilde{U}_{j,p}^{\Lambda}:\ \Lambda\ \text{closed orbit and}\ 
z_0=-\frac{2i\pi}{\ml{P}_{\Lambda}}\left(p+\varepsilon_{\Lambda}+\gamma_j^{\Lambda}\right)\right\},$$
which can be splitted as the direct sum
$$\bigoplus_{k=0}^nC^k_{V,\nabla}(z_0)=C^{\text{even}}_{V,\nabla}(z_0)\oplus C^{\text{odd}}_{V,\nabla}(z_0).$$
Recall from remark~\ref{r:chain-contraction-map} that $R$ is a chain contraction map for the complex $(C_{V,\nabla}^{\bullet}(z_0),d^{\nabla})$. 
Then, according to~\cite[p.~30]{F87}, the torsion of this finite dimensional complex (with respect to its preferred basis) is equal to
\begin{equation}\label{e:torsion-finite-dim}\tau\left(C_{V,\nabla}^{\bullet}(z_0),d^{\nabla}\right)
=\left|\det\left(d^{\nabla}+R\right)_{C^{\text{even}}\rightarrow C^{\text{odd}}}\right|,
\end{equation}
where the determinant is computed in the basis $(U_{j,p}^{\Lambda},\tilde{U}_{j,p}^{\Lambda})_{j,p,\Lambda}$. In order to make this computation, we will 
to prove a preliminary result:
\begin{lemm}\label{l:support} Let $(\Lambda,j,p)$ be such that
 $$z_0=-\frac{2i\pi}{\ml{P}_{\Lambda}}\left(p+\varepsilon_{\Lambda}+\gamma_j^{\Lambda}\right).$$
Then, the supports of $\ml{L}_{V,\nabla}^{-1}(U_{j,p}^{\Lambda})$ and $\ml{L}_{V,\nabla}^{-1}(\tilde{U}_{j,p}^{\Lambda})$ are equal to 
$\overline{W^u(\Lambda)}$.
\end{lemm}
\begin{proof} Recall from paragraph~\ref{sss:local-coordinates-closed-orbit} that the supports of $U_{j,p}^{\Lambda}$ and 
$\tilde{U}_{j,p}^{\Lambda}$ are equal to $\overline{W^u(\Lambda)}$.
According to Smale's Theorem~\ref{t:smale}, $\overline{W^u(\Lambda)}$ is the union of 
unstable manifolds $W^u(\Lambda')$ with $\text{dim} W^u(\Lambda')\leq \text{dim} W^u(\Lambda)$. As $\ml{L}_{V,\nabla}^{(k)}$ 
preserves $C^k_{V,\nabla}(z_0)$ and as the support of $\ml{L}_{V,\nabla}U_{j,p}^{\Lambda}$ is contained in $\overline{W^u(\Lambda)}$, we can write
$$\ml{L}_{V,\nabla}U_{j,p}^{\Lambda} = \alpha_{j,p,\Lambda}U_{j,p}^{\Lambda}
+\sum_{\Lambda'\leqq\Lambda, j',p':\text{dim} W^s(\Lambda')=k}\alpha_{\Lambda',j',p'}U_{\Lambda',j',p'},$$
where $\leqq$ is the partial order relation of Smale's Theorem~\ref{t:smale} and where $\alpha_*$ are complex numbers. Similarly,
 \begin{eqnarray*}\ml{L}_{V,\nabla}(\tilde{U}_{j,p}^{\Lambda}) &= &\tilde{\alpha}_{j,p,\Lambda}\tilde{U}_{j,p}^{\Lambda}
+\sum_{\Lambda'\leqq\Lambda, j',p':\text{dim} W^s(\Lambda')=k}\tilde{\alpha}_{\Lambda',j',p'}\tilde{U}^{\Lambda'}_{j',p'}\\
& &+\sum_{\Lambda'\leqq\Lambda, j',p':\text{dim} W^s(\Lambda')=k+1}\tilde{\alpha}_{\Lambda',j',p'}U^{\Lambda'}_{j',p'}, 
\end{eqnarray*}
 where $\alpha_{*}$ and $\tilde{\alpha}_*$ are complex numbers. From the explicit expressions 
 of $U_{j,p}^{\Lambda}$ and $\tilde{U}_{j,p}^{\Lambda}$, 
 we can already observe that 
 $\alpha_{j,p,\Lambda}=\tilde{\alpha}_{j,p,\Lambda}=-z_0$. As $z_0\neq 0$, we can also apply $\ml{L}_{V,\nabla}^{-1}$ to both sides of these equalities and we find that
$$U_{j,p}^{\Lambda} = -z_0 \ml{L}_{V,\nabla}^{-1}(U_{j,p}^{\Lambda})
+\sum_{\Lambda'\leqq\Lambda, j',p':\text{dim} W^s(\Lambda')=k}\alpha_{\Lambda',j',p'}\ml{L}_{V,\nabla}^{-1}(U^{\Lambda'}_{j',p'}),$$
 and
 \begin{eqnarray*}\tilde{U}_{j,p}^{\Lambda} &= &-z_0\ml{L}_{V,\nabla}^{-1}(\tilde{U}_{j,p}^{\Lambda})
+\sum_{\Lambda'\leqq\Lambda, j',p':\text{dim} W^s(\Lambda')=k}\tilde{\alpha}_{\Lambda',j',p'}\ml{L}_{V,\nabla}^{-1}(\tilde{U}^{\Lambda'}_{j',p'})\\
& &+\sum_{\Lambda'\leqq\Lambda, j',p':\text{dim} W^s(\Lambda')=k+1}\tilde{\alpha}_{\Lambda',j',p'}\ml{L}_{V,\nabla}^{-1}(U^{\Lambda'}_{j',p'}). 
\end{eqnarray*}
Thus, determining the support of $\ml{L}_{V,\nabla}^{-1}(U_{j,p}^{\Lambda})$ and of $\ml{L}_{V,\nabla}^{-1}(\tilde{U}_{j,p}^{\Lambda})$ 
follows from the fact that we are able to compute the support of 
their analogues for $\Lambda'\leqq\Lambda$ with $\text{dim} W^s(\Lambda')=k$ or $k+1$. Recall that the partial order relation $\leqq$ is defined 
in Theorem~\ref{t:smale} and that all the sums are implicitely over indices $(\Lambda',j',p')$ such that
\begin{equation}\label{e:relation-z0}z_0=-\frac{2i\pi}{\ml{P}_{\Lambda'}}\left(p'+\varepsilon_{\Lambda'}+\gamma_{j'}^{\Lambda'}\right).\end{equation}
If we consider the smallest such $\Lambda'$ and if we are able to prove that the corresponding supports are equal to $\overline{W^u(\Lambda')}$, 
then the lemma will follow from an induction argument 
on $\text{dim} W^u(\Lambda')$. Fix now $\Lambda'\leqq \Lambda$ which is minimal and $(j',p')$ such that~\eqref{e:relation-z0} is satisfied. 
As $\Lambda'$ is minimal, we can decompose 
$\ml{L}_{V,\nabla}U_{j',p'}^{\Lambda'}$ and $\ml{L}_{V,\nabla}\tilde{U}_{j',p'}^{\Lambda'}$ but there is now no remainder term, i.e. one has
$$\ml{L}_{V,\nabla}U_{j',p'}^{\Lambda'}=-z_0U_{j',p'}^{\Lambda'}\ \text{and}\ \ml{L}_{V,\nabla}\tilde{U}_{j',p'}^{\Lambda'}=-z_0\tilde{U}_{j',p'}^{\Lambda'}.$$
Applying $\ml{L}_{V,\nabla}^{-1}$ and using the fact $z_0\neq 0$, we can conclude that the supports of $\ml{L}_{V,\nabla}^{-1}U_{j',p'}^{\Lambda'}$ and 
$\ml{L}_{V,\nabla}^{-1}\tilde{U}_{j',p'}^{\Lambda'}$ are equal to $\overline{W^u(\Lambda')}$. By induction, we can then recover that the supports of 
$\ml{L}_{V,\nabla}^{-1}(U_{j,p}^{\Lambda})$ and of $\ml{L}_{V,\nabla}^{-1}(\tilde{U}_{j,p}^{\Lambda})$ are equal to $\overline{W^u(\Lambda)}$.
\end{proof}

We can now compute the torsion $(C_{V,\nabla}^{\bullet}(z_0),d^{\nabla})$ using formula~\eqref{e:torsion-finite-dim}. For that purpose, we 
fix $\Lambda$ to be a closed orbit of the flow such that $\text{dim}\ W^s(\Lambda)=k$. Suppose also that $(j,p)$ verifies the equality
$$z_0=-\frac{2i\pi}{\ml{P}_{\Lambda}}\left(p+\varepsilon_{\Lambda}+\gamma_j^{\Lambda}\right).$$
On the one hand, if $k$ is even, we would like to express $(d^{\nabla}+R)(\tilde{U}^{\Lambda}_{j,p})$ in the preferred basis of 
$C^{\text{odd}}_{V,\nabla}(z_0)$. 
On the other hand, if $k$ is odd, we are interested in the expression of $(d^{\nabla}+R)(U^{\Lambda}_{j,p})$. 
In order to compute these determinants, 
it will be convenient to order these basis which are indexed by $(\Lambda,j,p)$ 
according to Smale's partial order on the unstable manifolds. 
More precisely, we will order them in such a way that $(\Lambda',j',p')$ is less than $(\Lambda,j,p)$ whenever $\Lambda'\leqq\Lambda$.

\subsubsection{The case $k$ even} Let us start with 
the case $k\equiv 0\ \text{mod}\ 2$. Recall that near $\Lambda$, $\tilde{U}^{\Lambda}_{j,p}$ is of the form
$$\tilde{U}_{j,p}^{\Lambda}=e^{\frac{2i\pi (p+\varepsilon_{\Lambda})\theta}{\ml{P}_{\Lambda}}}
\left(P(\theta)^{-1}\right)^*\left(\delta_0(x_1,\ldots, x_{k-1})dx_1\wedge\ldots\wedge dx_{k-1}\right)\wedge d\theta \otimes \mathbf{c}_j^{\Lambda}(x,y,\theta).$$
First, from this expression, we can verify that $d^{\nabla}\tilde{U}_{j,p}^{\Lambda}$ is equal to $0$ in a neighborhood of $\Lambda$. Moreover, as 
it satisfies $\ml{L}_{V,\nabla}d^{\nabla}\tilde{U}_{j,p}^{\Lambda}$ is equal to $0$ on $W^u(\Lambda)$ and as $d^{\nabla}\tilde{U}_{j,p}^{\Lambda}$ 
is supported on $\overline{W^u(\Lambda)}$, we can deduce that $d^{\nabla}\tilde{U}_{j,p}^{\Lambda}$ is supported on 
$\overline{W^u(\Lambda)}-W^u(\Lambda)$. From Smale's Theorem~\cite{Sm60} and from the fact that $(C_{V,\nabla}^{\bullet}(z_0),d^{\nabla})$ 
is a cohomological complex, we can deduce that
$$d^{\nabla}\tilde{U}_{j,p}^{\Lambda}=\sum_{\Lambda'\leqq\Lambda, j',p':\text{dim} W^s(\Lambda')=k+2}\alpha_{\Lambda',j',p'}U^{\Lambda'}_{j',n'}
+\sum_{\Lambda'\leqq\Lambda, j',p':\text{dim} W^s(\Lambda')=k+1}\alpha_{\Lambda',j',p'}\tilde{U}^{\Lambda'}_{j',p'}.$$
where we only sum over the $(\Lambda',j',p')$ satisfying $$z_0=-\frac{2i\pi}{\ml{P}_{\Lambda'}}\left(p'+\varepsilon_{\Lambda'}+\gamma_{j'}^{\Lambda'}
\right).$$
Now, we can compute $R\tilde{U}_{j,p}^{\Lambda}$. First, we write that $\ml{L}_{V,\nabla}(\tilde{U}_{j,p}^{\Lambda})$ belongs to 
$C^k_{V,\nabla}(z_0)$ and that its support is contained in $\overline{W^u(\Lambda)}$. Hence, one can write
\begin{eqnarray*}\ml{L}_{V,\nabla}(\tilde{U}_{j,p}^{\Lambda}) &= &\alpha_{j,n,\Lambda}\tilde{U}_{j,p}^{\Lambda}
+\sum_{\Lambda'\leqq\Lambda, j',p':\text{dim} W^s(\Lambda')=k}\alpha_{\Lambda',j',p'}\tilde{U}^{\Lambda'}_{j',n'}\\
& &+\sum_{\Lambda'\leqq\Lambda, j',p':\text{dim} W^s(\Lambda')=k+1}\alpha_{\Lambda',j',p'}U^{\Lambda'}_{j',p'}. 
\end{eqnarray*}
From the expression of $\tilde{U}_{j,p}^{\Lambda}$ near $\Lambda$, we know that $\alpha_{\Lambda,j,p}=-z_0$. Hence, we find
\begin{eqnarray*}\ml{L}_{V,\nabla}^{-1}(\tilde{U}_{j,p}^{\Lambda}) &= &-\frac{1}{z_0}\tilde{U}_{j,p}^{\Lambda}
-\sum_{\Lambda'\leqq\Lambda, j',p':\text{dim} W^s(\Lambda')=k}\alpha_{\Lambda',j',p'}\ml{L}_{V,\nabla}^{-1}\tilde{U}^{\Lambda'}{j',p'}\\
& &-\sum_{\Lambda'\leqq\Lambda, j',p':\text{dim} W^s(\Lambda')=k+1}\alpha_{\Lambda',j',p'}\ml{L}_{V,\nabla}^{-1}U^{\Lambda'}_{j',p'}. 
\end{eqnarray*}
Recall now that $R=\iota_V\ml{L}_{V,\nabla}^{-1}$ when it acts on $C^{\bullet}_{V,\nabla}(z_0)$. Hence, one has
\begin{eqnarray*}R(\tilde{U}_{j,p}^{\Lambda}) &= &-\frac{1}{z_0}\iota_V(\tilde{U}_{j,p}^{\Lambda})
-\sum_{\Lambda'\leqq\Lambda, j',p':\text{dim} W^s(\Lambda')=k}\alpha_{\Lambda',j',p'}\iota_V\left(\ml{L}_{V,\nabla}^{-1}\tilde{U}^{\Lambda'}_{j',p'}\right)\\
& &-\sum_{\Lambda'\leqq\Lambda, j',p':\text{dim} W^s(\Lambda')=k+1}\alpha_{\Lambda',j',p'}\iota_V\left(\ml{L}_{V,\nabla}^{-1}U^{\Lambda'}_{j',p'}\right). 
\end{eqnarray*}
Arguing as in the proof of Proposition~\ref{p:koszul} -- see paragraph~\ref{ss:koszul-nonzero}, we can verify 
that $\iota_V(\tilde{U}_{j,p}^{\Lambda})$ is equal to $U_{j,p}^{\Lambda}$ plus some remainder term which 
is carried on $\overline{W^u(\Lambda)}-W^u(\Lambda)$. Hence to summarize, one finds that
\begin{equation}\label{e:torsion-even}(d^{\nabla}+R)(\tilde{U}_{j,p}^{\Lambda})=-\frac{1}{z_0}U_{j,p}^{\Lambda}+T_{j,p}^{\Lambda},\end{equation}
with $T_{j,p}^{\Lambda}$ belonging to $C^{\text{odd}}_{V,\nabla}(z_0)$ and supported in $\overline{W^u(\Lambda)}-W^u(\Lambda)$.

\subsubsection{The case $k$ odd} Suppose now that $k\equiv 1\ \text{mod}\ 2$. Near $\Lambda$, $U^{\Lambda}_{j,p}$ is of the form
$$U_{j,p}^{\Lambda}=e^{\frac{2i\pi (p+\varepsilon_{\Lambda})\theta}{\ml{P}_{\Lambda}}}
\left(P(\theta)^{-1}\right)^*\left(\delta_0(x_1,\ldots, x_{k-1})dx_1\wedge\ldots\wedge dx_{k-1}\right) \otimes \mathbf{c}_j^{\Lambda}(x,y,\theta).$$
From paragraph~\ref{ss:koszul-nonzero}, we know that $\iota_V(U_{j,p}^{\Lambda})=0$. Hence, from the definition of $R$, one finds that 
$(d^{\nabla}+R)(U_{j,p}^{\Lambda})=d^{\nabla}(U_{j,p}^{\Lambda})$. Applying $d^{\nabla}$ to $U_{j,p}^{\Lambda}$, we find that, in a neighborhood of $\Lambda$, it is equal 
to $z_0 \tilde{U}_{j,p}^{\Lambda}$ (recall from paragraph~\ref{sss:local-coordinates-closed-orbit} that $\nabla \mathbf{c}_j^{\Lambda}=\frac{2i\pi\gamma_j^{\Lambda}}{\ml{P}_{\Lambda}}\mathbf{c}_j^{\Lambda}$). 
The current $d^{\nabla}(U_{j,p}^{\Lambda})-z_0\tilde{U}_{j,p}^{\Lambda}$ belongs to $C^{k}_{V,\nabla}(z_0)$, it is supported inside 
$\overline{W^u(\Lambda)}$ and it identically vanishes near $\Lambda$. Hence, by propagation, it is supported inside $\overline{W^u(\Lambda)}-W^u(\Lambda)$. Equivalently, one has
\begin{equation}\label{e:torsion-odd}(d^{\nabla}+R)(U_{j,p}^{\Lambda})= z_0 \tilde{U}_{j,p}^{\Lambda}+T_{j,p}^{\Lambda},\end{equation}
with $T_{j,p}^{\Lambda}$ belonging to $C^{\text{odd}}_{V,\nabla}(z_0)$ and supported in $\overline{W^u(\Lambda)}-W^u(\Lambda)$.

\subsubsection{Conclusion}

Combining~\eqref{e:torsion-finite-dim} with~\eqref{e:torsion-even} and~\ref{e:torsion-odd}, one finally finds the following expression for the torsion:
\begin{equation}\label{e:torsion-z0}\tau\left(C_{V,\nabla}^{\bullet}(z_0),d^{\nabla}\right)
=\prod_{(\Lambda,j,p):(*)}\left|\frac{2\pi(n+\varepsilon_{\Lambda}+\gamma_j^{\Lambda})}{\ml{P}_{\Lambda}}\right|^{(-1)^{n+\text{dim}\ W^u(\Lambda)}},
\end{equation}
where $(*)$ means that we take the product over the triples $(\Lambda,j,p)$ satisfying
$$z_0=-\frac{2i\pi}{\ml{P}_{\Lambda}}\left(p+\varepsilon_{\Lambda}+\gamma_j^{\Lambda}\right).$$
In other words, the torsion $\tau\left(C_{V,\nabla}^{\bullet}(z_0),d^{\nabla}\right)$ is of the form $|z_0|^{(-1)^nm(z_0)}$ with 
$$m(z_0)=\sum_{\Lambda\ \text{closed orbit}}(-1)^{n+\text{dim}\ W^u(\Lambda)}\left|\left\{(p,j):z_0=-\frac{2i\pi}{\ml{P}_{\Lambda}}\left(p+\varepsilon_{\Lambda}+\gamma_j^{\Lambda}\right)\right\}\right|.$$
In a more spectral manner, we have, using the results of paragraphs~\ref{ss:koszul-nonzero},
\begin{equation}\label{e:torsion-subspace}\ln\ \tau\left(C_{V,\nabla}^{\bullet}(z_0),d^{\nabla}\right)=\ln |z_0| \sum_{k=0}^n(-1)^{k+1}\text{dim}\left( C^k_{V,\nabla}(z_0)\cap\text{Ker}(\iota_V)\right).
\end{equation}
Observe that, if we set
\begin{eqnarray*}\zeta_{z_0}(s) &:= &\sum_{(\Lambda,j,p):(*)}(-1)^{n+\text{dim}\ W^u(\Lambda)}
\left|\frac{2\pi(p+\varepsilon_{\Lambda}+\gamma_j^{\Lambda})}{\ml{P}_{\Lambda}}\right|^{-s}\\ 
& = &\frac{1}{|z_0|^s}\sum_{k=0}^n(-1)^{k+1}\text{dim}\left( C^k_{V,\nabla}(z_0)\cap\text{Ker}(\iota_V)\right),\end{eqnarray*}
then one easily finds that
\begin{equation}\label{e:zeta-trivial}\tau\left(C_{V,\nabla}^{\bullet}(z_0),d^{\nabla}\right)=e^{-\zeta_{z_0}'(0)},\end{equation}
which motivates the upcoming definitions.

\subsection{Another zeta function associated to the flow.}\label{ss:zeta-torsion} For every $0\leq k\leq n$, we define the infinite dimensional vector space of Pollicott-Ruelle resonant states:
$$C^{k}_{V,\nabla}(i\IR^*)=\bigoplus_{z_0\in \ml{R}_k(V,\nabla)\cap i\IR^*}C^k_{V,\nabla}(z_0).$$
This induces an infinite dimensional complex $(C^{\bullet}_{V,\nabla}(i\IR^*),d^{\nabla})$ and we define the following zeta function:
\begin{equation}\label{e:zeta-dynamic}\zeta_{V,\nabla}(s):=\sum_{\Lambda,j,p}(-1)^{n+\text{dim}\ W^u(\Lambda)}\left|\frac{\ml{P}_{\Lambda}}{2\pi(p+\varepsilon_{\Lambda}+\gamma_j^{\Lambda})}\right|^{s},\end{equation}
where the sum runs over the closed orbits $\Lambda$ of the flow, $1\leq j\leq N$ and $p\in\IZ$ with the assumption that 
$p+\varepsilon_{\Lambda}+\gamma_j^{\Lambda}\neq 0$. In a more spectral manner, this can be written
$$\zeta_{V,\nabla}(s)=\sum_{k=0}^n(-1)^{k+1}\sum_{z_0\in\ml{R}_k(V,\nabla)\cap i\IR^*}\frac{1}{|z_0|^s}\text{dim}\left( C^k_{V,\nabla}(z_0)\cap\text{Ker}(\iota_V)\right).$$
We shall explain below by classical arguments from Hodge theory~\cite{Mn14} that, up to the modulus, this zeta function is related to the spectral zeta function $\zeta_{RS}$ 
from corollary~\ref{c:Spectralzetaruelle}. Yet, we emphasize that \emph{it appears here really from 
the computation of the torsion of an acyclic complex in a preferred basis} and not just as a reproduction of Ray-Singer's definition for the Laplace operator. In any case, 
this function is well defined 
for $\text{Re}(s)>1$ and we aim at describing its meromorphic extension to $\IC$. In particular, motivated by~\eqref{e:zeta-trivial}, 
we would like to define the torsion of the cohomological complex $(C^{\bullet}_{V,\nabla}(i\IR^*),d^{\nabla})$
as $e^{-\zeta_{V,\nabla}'(0)}$ provided that it makes sense.


Let us now study the meromorphic continuation of the zeta functions we have just defined. For that purpose, we write
$$\zeta_{V,\nabla}(s)=\sum_{\Lambda,1\leq j\leq N}(-1)^{n+\text{dim}\ W^u(\Lambda)}\left(\frac{1}{2\pi}\right)^s
\sum_{p\in\IZ:\ p+\varepsilon_{\Lambda}+\gamma_j^{\Lambda}\neq 0}\left|\frac{\ml{P}_{\Lambda}}{(p+\varepsilon_{\Lambda}+\gamma_j^{\Lambda})}\right|^{s}.$$
Hence, equivalently, it amounts to understand the meromorphic continuation of
$$\xi_{j,\Lambda}(s):=\sum_{p\in\IZ:\ p+\varepsilon_{\Lambda}+\gamma_j^{\Lambda}\neq 0}\left|\frac{1}{(p+\varepsilon_{\Lambda}+\gamma_j^{\Lambda})}\right|^{s},$$
for every closed orbit and for every $1\leq j\leq N$. This can in fact easily be rewritten in terms of Hurwitz zeta functions and of Riemann zeta functions. Recall that 
$\gamma_j^{\Lambda}\in[0,1)$ and that $\varepsilon_{\Lambda}\in\{0,1/2\}$ (the value depending on the orientability of $W^u(\Lambda)$). In particular, there exists a unique 
$q_j^{\Lambda}\in(0,1]$ such that $\varepsilon_{\Lambda}+\gamma_j^{\Lambda}$ is equal to $q_j^{\Lambda}$ modulo $1$. Introduce now the so-called Hurwitz zeta 
function~\cite[Ch.~12]{Ap76}:
$$\forall q\in (0,1],\ \zeta(s,q):=\sum_{n=0}^{+\infty}\frac{1}{(n+q)^s}.$$
Hence, for $q_j^{\Lambda}\neq 1$, one has
$$\xi_{j,\Lambda}(s)=\zeta(s,q_j^{\Lambda})+\zeta(s,1-q_j^{\Lambda}),$$
while, for $q_j^{\Lambda}= 1$, 
$$\xi_{j,\Lambda}(s)=2\zeta(s,1),$$
which is nothing else than twice the Riemann zeta function. From~\cite[Th.~12.4-5]{Ap76}, we can conclude that, for every $1\leq j\leq N$ and for every 
closed orbit $\Lambda$, $\xi_{j,\Lambda}(s)$ has a meromorphic extension to $\IC$ which is analytic except for a simple pole at $s=1$ (with residue $2$). In particular, 
$\zeta_{V,\nabla}(s)$ extends meromorphically to $\IC$ with a unique pole at $s=1$ which is simple and whose residue is given by 
$$\frac{N}{\pi}\sum_{\Lambda\ \text{closed orbit}}(-1)^{n+\text{dim}\ W^u(\Lambda)}\ml{P}_{\Lambda}.$$

\subsection{The regularized torsion of the infinite dimensional complex $ (C^{\bullet}_{V,\nabla}(i\IR^*),d^{\nabla})$.}
\label{ss:reg-torsion} 
We define the \textbf{regularized torsion} of the infinite dimensional complex 
$(C^{\bullet}_{V,\nabla}(i\IR^*),d^{\nabla})$ as follows:
$$T\left(C^{\bullet}_{V,\nabla}(i\IR^*),d^{\nabla}\right):=e^{-\zeta_{V,\nabla}'(0)},$$
which depends implicitely on the choice of the basis from paragraph~\ref{sss:local-coordinates-closed-orbit}. Let us compute the contribution coming from 
the derivative of the zeta function. We have
$$\zeta_{V,\nabla}'(s)=(-1)^n\sum_{\Lambda,1\leq j\leq N}(-1)^{\text{dim}\ W^u(\Lambda)}\left(\frac{1}{2\pi}\right)^s\left(\xi_{j,\Lambda}'(s)
+\ln\left(\frac{\ml{P}_{\Lambda}}{2\pi}\right)\xi_{j,\Lambda}(s)\right).$$
In order to compute the derivative at $s=0$, we shall come back to the expression in terms of Hurwitz zeta function. According to~\cite[Th.~12.13]{Ap76}, we know that
$\zeta(0,q)=\frac{1}{2}-q$. Hence, for $q_j^{\Lambda}=1$, one has $\xi_{j,\Lambda}(0)=-1$ while, for $q_j^{\Lambda}\neq 1$, $\xi_{j,\Lambda}(0)=0$. For 
the derivatives, one has $\xi_{j,\Lambda}'(0)=-\ln(2\pi)$ if $q_j^{\Lambda}=1$ and 
$\xi_{j,\Lambda}'(0)=-\ln(2\sin(\pi q_j^{\Lambda}))$ otherwise~\cite[p.~271]{WhWa96}~\cite[p.~195]{Ni08}.

Recall that we used the notations $\Delta_{\Lambda}=e^{2i\varepsilon_{\Lambda}\pi}$ and $M_{\ml{E}}(\Lambda)$ for the monodromy matrix 
around $\Lambda$. Also, we have set $q_j^{\Lambda}=\varepsilon_{\Lambda}+\gamma_j^{\Lambda}\ \text{mod}\ 1$ where $e^{2i\pi\gamma_j^{\Lambda}}$ are the eigenvalues 
of the monodromy matrix $M_{\ml{E}}(\Lambda)$. With these notations, we finally find
$$\zeta_{V,\nabla}'(0)=-\sum_{\Lambda,1\leq j\leq N: q_j^{\Lambda}\neq 1}(-1)^{n+\text{dim}\ W^u(\Lambda)}\ln\left(2\sin(\pi q_j^{\Lambda})\right)
-\sum_{\Lambda}(-1)^{n+\text{dim}\ W^u(\Lambda)}m_{\Lambda}\ln\ml{P}_{\Lambda},$$
and thus
$$e^{-\zeta_{V,\nabla}'(0)}=\prod_{\Lambda,1\leq j\leq N: \Delta_{\Lambda}e^{2i\pi\gamma_j^{\Lambda}}\neq 1}
\left|1-\Delta_{\Lambda}e^{2i\pi\gamma_j^{\Lambda}}\right|^{(-1)^{n+\text{dim}\ W^u(\Lambda)}}
\prod_{\Lambda}\ml{P}_{\Lambda}^{(-1)^{n+\text{dim}\ W^u(\Lambda)}m_{\Lambda}}.$$
Note that if $\Delta_{\Lambda}$ is not an eigenvalue of $M_{\ml{E}}(\Lambda)$ for every closed orbit\footnote{This is the assumption made in~\cite{F87}.}, then we find
$$e^{-\zeta_{V,\nabla}'(0)}=\prod_{\Lambda}
\left|\det\left(\text{Id}-\Delta_{\Lambda}M_{\ml{E}}(\Lambda)\right)\right|^{(-1)^{n+\text{dim}\ W^u(\Lambda)}}
,$$
which is equal to the Reidemester torsion (if $V$ is also nonsingular) thanks to the work of Fried in~\cite[Sect.~3]{F87}. Recall that the left hand side can be interpreted 
either in terms of the nonzero Pollicott-Ruelle resonances on the imaginary axis or in terms of the torsion of the corresponding resonant states. To summarize, we have shown

\begin{prop}\label{c:RS} Suppose that the assumptions of Theorem~\ref{t:fried0} are satisfied. Then, one has 
\begin{itemize} 
\item $\zeta_{V,\nabla}(s)$ has a meromorphic 
extension to $\IC$ with a unique pole at $s=1$ which is simple and whose residue equals
 $$\frac{N}{\pi}\sum_{\Lambda\ \text{closed orbit}}(-1)^{n+\operatorname{dim}\ W^u(\Lambda)}\ml{P}_{\Lambda}.$$
\item one has
$$e^{(-1)^n\zeta_{V,\nabla}'(0)}=\prod_{\Lambda,1\leq j\leq N: \Delta_{\Lambda}e^{2i\pi\gamma_j^{\Lambda}}\neq 1}
\left|1-\Delta_{\Lambda}e^{2i\pi\gamma_j^{\Lambda}}\right|^{(-1)^{\operatorname{dim}\ W^u(\Lambda)+1}}
\prod_{\Lambda}\ml{P}_{\Lambda}^{(-1)^{\operatorname{dim}\ W^u(\Lambda)+1}m_{\Lambda}},$$
where $(e^{2i\pi\gamma_j^{\Lambda}})_{j=1,\ldots, N}$ are the eigenvalues of $M_{\ml{E}}(\Lambda)$. 
\item if $V$ is nonsingular and $m_{\Lambda}=0$ for every closed orbit, then $e^{(-1)^n\zeta_{V,\nabla}'(0)}$ is equal to the Reidemeister torsion of 
$(\ml{E},\nabla)$.
\end{itemize}
\end{prop}

\subsection{Twisted Fuller measures}\label{s:fried}

In this paragraph, we prove Theorem~\ref{t:fried0}. First of all, 
we rewrite the Fuller measure in terms of the dimension of the unstable manifolds of the critical elements and the twisting 
index $\Delta_\Lambda$ of the closed orbits, i.e.
\begin{eqnarray*}\mu_{V,\nabla}(t) & =&-\frac{N}{t}\sum_{\Lambda\ \text{fixed point}}(-1)^{\text{dim}\ W^u(\Lambda)}\\
& -& \sum_{\Lambda\ \text{closed orbit}}
\sum_{m\geq 1}\frac{1}{m}(-1)^{\text{dim}\ W^u(\Lambda)}\Delta_{\Lambda}^m\text{Tr}\left(M_{\ml{E}}(\Lambda)^m\right) \delta(t-m\ml{P}_{\Lambda}).
\end{eqnarray*}

The Morse inequality from Corollary~\ref{c:morse-smale-inequality}  turns out to be an \textbf{equality} in the case $k=n$, 
thus we find that $\sum_{\Lambda\ \text{fixed point}}(-1)^{\text{dim}\ W^u(\Lambda)}=\sum_k (-1)^kb_k(M,\cE) =\chi(M,\cE) $. Therefore the above formula can be rewritten in a more compact way
as~:
\begin{equation}\label{e:trace-guillemin}\mu_{V,\nabla}(t)=-\frac{\chi(M,\ml{E})}{t}-\sum_{\Lambda\ \text{closed orbit}}
\sum_{m\geq 1}\frac{1}{m}(-1)^{\text{dim}\ W^u(\Lambda)}\Delta_{\Lambda}^m\text{Tr}\left(M_{\ml{E}}(\Lambda)^m\right) \delta(t-m\ml{P}_{\Lambda}).\end{equation}
More explicitely, we can also write the sum on the r.h.s of the above identity in terms 
of the eigenvalues $\left(e^{ 2i\pi \gamma_j^\Lambda }\right)_{j=1}^N $ of the monodromy matrices
and the coefficients $\varepsilon_\Lambda$ such that $\Delta_\Lambda=e^{2i\pi\varepsilon_\Lambda}$
which yields~:
$$t\mu_{V,\nabla}(t)=-N\sum_{\Lambda\ \text{fixed point}}(-1)^{\text{dim}\ W^u(\Lambda)}-\sum_{\Lambda,j}\ml{P}_{\Lambda}(-1)^{\text{dim}\ W^u(\Lambda)}
\sum_{m\geq 1}\left(e^{2i\pi\frac{\gamma_j^{\Lambda}+\varepsilon_{\Lambda}}{\ml{P}_{\Lambda}}}\right)^{m\ml{P}_{\Lambda}} \delta(t-m\ml{P}_{\Lambda}).$$
For $T>0$ and $\gamma\in\IR$, recall that the Poisson
summation formula 
implies that~:
$$\sum_{m\in\IZ}\delta(t-mT) e^{2i\pi\gamma mT}=\frac{1}{T}\sum_{l\in\IZ}e^{\frac{2i\pi}{T}(l+\gamma T)t} .$$
We can now apply the above formula for $T=\ml{P}_\Lambda$ and $\gamma=\frac{\gamma_j^{\Lambda}+\varepsilon_{\Lambda}}{\ml{P}_{\Lambda}}$,
which gives us~:
$$t\mu_{V,\nabla}(t)=-N\sum_{\Lambda\ \text{fixed point}}(-1)^{\text{dim}\ W^u(\Lambda)}-\sum_{\Lambda,j}(-1)^{\text{dim}\ W^u(\Lambda)}
\sum_{l\in\IZ}e^{2i\pi t\frac{l+\gamma_j^{\Lambda}+\varepsilon_{\Lambda}}{\ml{P}_{\Lambda}}},$$
in the sense of distributions in $\ml{D}'(\IR_+^*)$. 

Finally, this quantity can be rewritten in a more spectral manner thanks to the results of paragraphs~\ref{ss:proof-koszul} 
and~\ref{ss:koszul-nonzero}.
More precisely,
\begin{equation}\label{e:beautiful}
\boxed{t\mu_{V,\nabla}(t)=\sum_{k=0}^n(-1)^{n-k+1}\sum_{z_0\in\ml{R}_k(V,\nabla)\cap i\IR}\text{dim} \left(C_{V,\nabla}^k(z_0)\cap\text{Ker}(\iota_V)\right)
e^{ z_0 t},}
\end{equation}
which concludes the proof of Theorem~\ref{t:fried0}.

\subsection{Fried's torsion functions}\label{ss:hodge-torsion}

It now remains to prove Corollary~\ref{c:Spectralzetaruelle} related to the torsion function. Recall that 
we introduced in subsection~\ref{ss:mainresultstorsion} the following zeta function~:
\begin{eqnarray*}
\zeta_{V,\nabla}^{\flat}(s,z)=\frac{1}{\Gamma(s)}\int_{0}^{+\infty} e^{-tz}t\mu_{V,\nabla}(t) t^{s-1}dt,
\end{eqnarray*}
which, by definition of the twisted Fuller measure, can be rewritten thanks to~\eqref{e:trace-guillemin} as
$$\zeta_{V,\nabla}^{\flat}(s,z)=-\frac{\chi(M,\ml{E})}{z^s}-\frac{1}{\Gamma(s)}\sum_{\Lambda\ \text{closed orbit}}(-1)^{\text{dim}\ W^u(\Lambda)}
\ml{P}_{\Lambda}^s
\sum_{m\geq 1} \text{Tr}\left(\left(e^{-\ml{P}_{\Lambda }z}\Delta_{\Lambda}M_{\ml{E}}(\Lambda)\right)^{m}\right) m^{s-1}.$$
Recall that $\Gamma(s)^{-1}=s+o(s)$. Hence, if we differentiate $\zeta_{V,\nabla}^{\flat}$ with respect to $s$, then we find that, for $\text{Re}(z)$ large enough,
$$\partial_s\zeta_{V,\nabla}^{\flat}(0,z)=\chi(M,\ml{E})\log z-\sum_{\Lambda\ \text{closed orbit}}(-1)^{\text{dim}\ W^u(\Lambda)}
\sum_{m\geq 1} \frac{\text{Tr}\left(\left(e^{-\ml{P}_{\Lambda }z}\Delta_{\Lambda}M_{\ml{E}}(\Lambda)\right)^{m}\right)}{m}$$
where one has 
$$\sum_{m\geq 1} \frac{\text{Tr}\left(\left(e^{-\ml{P}_{\Lambda }z}\Delta_{\Lambda}M_{\ml{E}}(\Lambda)\right)^{m}\right)}{m}=-\log\text{det}\left(\text{Id}-e^{-\ml{P}_{\Lambda }z}\Delta_{\Lambda}M_{\ml{E}}(\Lambda) \right)$$
as soon as $\text{Re}(z)>0$ since $\Vert e^{-\ml{P}_{\Lambda }z}\Delta_{\Lambda}M_{\ml{E}}(\Lambda)\Vert\leqslant e^{-\ml{P}_{\Lambda }Re(z)}<1$ by unitarity of $\Delta_{\Lambda}M_{\ml{E}}(\Lambda)$.

Equivalently, this can be rewritten as
$$\partial_s\zeta_{V,\nabla}^{\flat}(0,z)=\chi(M,\ml{E})\log z+\sum_{\Lambda\ \text{closed orbit}}(-1)^{\text{dim}\ W^u(\Lambda)}
\log\text{det}\left(\text{Id}-e^{-\ml{P}_{\Lambda }z}\Delta_{\Lambda}M_{\ml{E}}(\Lambda)\right).$$
Note that this expression is well defined for $\text{Re}(z)>0$ and since the torsion function $Z_{V,\nabla}$ was defined as
$Z_{V,\nabla}(z)=e^{-\partial_s\zeta_{V,\nabla}^{\flat}(0,z)}$, we deduce the following identity relating the torsion function and the weighted Ruelle zeta function~:
$$Z_{V,\nabla}(z)=z^{-\chi(M,\ml{E})} \prod_{\Lambda\ \text{closed orbit}}
\text{det}\left(\text{Id}-e^{-\ml{P}_{\Lambda }z}\Delta_{\Lambda}M_{\ml{E}}(\Lambda)\right)^{-(-1)^{\text{dim}\ W^u(\Lambda)}},$$
as was stated in paragraph~\ref{ss:mainresultstorsion}.

In the case where $V$ is a nonsingular Morse-Smale vector field, we recognize the torsion function introduced by Fried in~\cite[p.~51-53]{F87}, also called twisted Ruelle zeta function. 
In any case, we can already verify from the exact expressions of the eigenvalues and their multiplicities (see paragraphs~\ref{sss:local-coordinates-critical-point} 
and~\ref{sss:local-coordinates-closed-orbit}) that the poles and zeros of $Z_{V,\nabla}(z)$ are completely determined by the resonances on the imaginary axis.

Let us now come back to Corollary~\ref{c:Spectralzetaruelle} and perform the Mellin transform on the right hand side of~\eqref{e:beautiful}. We find that
$$\zeta_{V,\nabla}^{\flat}(s,z)=\frac{1}{\Gamma(s)}\sum_{k=0}^n(-1)^{n-k+1}\sum_{z_0\in\ml{R}_k(V,\nabla)\cap i\IR}\text{dim} \left(C_{V,\nabla}^k(z_0)\cap\text{Ker}(\iota_V)\right)
\int_0^{+\infty}e^{ (z_0-z) t}t^{s-1}dt.$$
This is also equal to
$$\zeta_{V,\nabla}^{\flat}(s,z)=\sum_{k=0}^n(-1)^{n-k+1}\sum_{z_0\in\ml{R}_k(V,\nabla)\cap i\IR}\text{dim} \left(C_{V,\nabla}^k(z_0)\cap\text{Ker}(\iota_V)\right)
(z-z_0)^{-s}.$$

\begin{rema}
\label{r:formal}

This expression can be \emph{formally} differentiated at $s=0$ yielding~:
$$\partial_s\zeta_{V,\nabla}^{\flat}(0,z)=-\sum_{k=0}^n(-1)^{n-k+1}\sum_{z_0\in\ml{R}_k(V,\nabla)\cap i\IR}\text{dim} \left(C_{V,\nabla}^k(z_0)\cap\text{Ker}(\iota_V)\right)
\log(z-z_0),$$
from which we deduce the \textbf{formal} expression relating
Fried's torsion and some infinite product indexed by the resonances 
lying on the imaginary axis~:
$$Z_{V,\nabla}(z)=\prod_{k=0}^n\prod_{z_0\in\ml{R}_k(V,\nabla)\cap i\IR}
(z-z_0)^{(-1)^{n-k+1}\text{dim} \left(C_{V,\nabla}^k(z_0)\cap\text{Ker}(\iota_V)\right)}.$$
\end{rema}

Using our expression for the kernel of $\mathcal{L}_{V,\nabla}$, we now find that
\begin{eqnarray*}\zeta_{V,\nabla}^{\flat}(s,z) & = &-\left(\chi(M,\ml{E})+\sum_{\Lambda\ \text{closed orbit}}(-1)^{\text{dim}\ W^u(\Lambda)}m_{\Lambda} \right)z^{-s}\\
& +& \sum_{k=0}^n(-1)^{n-k+1}\sum_{z_0\in\ml{R}_k(V,\nabla)\cap i\IR^*}\text{dim} \left(C_{V,\nabla}^k(z_0)\cap\text{Ker}(\iota_V)\right)
(z-z_0)^{-s}.\end{eqnarray*}
Now, in the spirit of Ray-Singer definition of analytic torsion, we can rewrite the second term in the right-hand side in a slightly different manner.
As in the case of Hodge theory, we can observe that
$$\text{dim}\ C_{V,\nabla}^k(z_0)=\text{dim} \left(C_{V,\nabla}^k(z_0)\cap\text{Ker}(\iota_V)\right)+
\text{dim} \left(C_{V,\nabla}^k(z_0)\cap\text{Ker}(d^{\nabla})\right)$$
and that
$$\text{dim} \left(C_{V,\nabla}^k(z_0)\cap\text{Ker}(\iota_V)\right)=\text{dim} \left(C_{V,\nabla}^{k+1}(z_0)\cap\text{Ker}(d^{\nabla})\right).$$
As in~\cite[paragraph 8.2 p.~53]{Mn14}, this implies that
$$\zeta_{V,\nabla}^{\flat}(s,z)=-\left(\chi(M,\ml{E})+\sum_{\Lambda\ \text{closed orbit}}(-1)^{\text{dim}\ W^u(\Lambda)}m_{\Lambda} \right)z^{-s}+\zeta_{RS}(s,z),$$
as expected.

\begin{rema}
Following Dyatlov and Zworski~\cite{DyZw13}, we could also have defined the (twisted) dynamical zeta function as~:
\begin{eqnarray*}
\tilde{\zeta}_{V,\nabla}^{\flat}(z):=\int_{0}^{+\infty} e^{-tz}t\mu_{V,\nabla}(t) dt,
\end{eqnarray*}
which can be rewritten thanks to~\eqref{e:trace-guillemin} as
$$\tilde{\zeta}_{V,\nabla}^{\flat}(z)=-\frac{\chi(M,\ml{E})}{z}-\sum_{\Lambda\ \text{closed orbit}}
(-1)^{\text{dim}\ W^u(\Lambda)}\ml{P}_{\Lambda}\sum_{m\geq 1}\text{Tr}\left(\Delta_{\Lambda}M_{\ml{E}}(\Lambda)e^{-z\ml{P}_{\Lambda}}\right)^m,$$
or equivalently
$$\tilde{\zeta}_{V,\nabla}^{\flat}(z)=-\frac{\chi(M,\ml{E})}{z}-\sum_{\Lambda\ \text{closed orbit}}
(-1)^{\text{dim}\ W^u(\Lambda)}
\text{Tr}\left(\frac{\ml{P}_{\Lambda}\Delta_{\Lambda}M_{\ml{E}}(\Lambda)e^{-z\ml{P}_{\Lambda}}}{\text{Id}-\Delta_{\Lambda}M_{\ml{E}}(\Lambda)e^{-z\ml{P}_{\Lambda}}}\right).$$
In that case, we would find that the poles of this zeta function are contained inside the intersection of the Pollicott-Ruelle resonances with 
the imaginary axis.
\end{rema}


\appendix

\section{A brief reminder on Anosov and Morse-Smale flows}\label{a:flows}

In this appendix, we briefly review some classical definitions and results from dynamical systems.

\subsection{Morse-Smale flows}

We say that $\Lambda\subset M$ is an elementary critical element if $\Lambda$ is either a fixed point or a closed orbit of $\varphi^t$. 
Such an element is said to be hyperbolic if the fixed point or the closed orbit is hyperbolic -- see~\cite{AbMa08} or the appendix of~\cite{DaRi17a} for a brief reminder. 
Following~\cite[p.~798]{Sm67}, $\varphi^t$ is \textbf{a Morse-Smale flow} if the following properties hold:
\begin{enumerate}
 \item the non-wandering set $\operatorname{NW}(\varphi^t)$ is the union of finitely many elementary critical element $\Lambda_1,\ldots,\Lambda_K$ 
 which are hyperbolic,
 \item for every $i,j$ and for every $x$ in $W^u(\Lambda_j)\cap W^s(\Lambda_i)$, one has
 \footnote{See appendix of~\cite{DaRi17a} for the precise definition of the stable/unstable manifolds $W^{s/u}(\Lambda)$.} 
 $T_xM=T_xW^u(\Lambda_j)+T_xW^s(\Lambda_i)$.
\end{enumerate}
 We now briefly expose some important properties of Morse-Smale flows and we refer to~\cite{Fr82, PaDeMe82, DaRi17a} for a more 
detailed exposition on the dynamical properties of these flows.
Under such assumptions, one can show that, for every $x$ in $M$, there exists a unique couple $(i,j)$ such 
that $x\in W^u(\Lambda_j)\cap W^s(\Lambda_i)$ (see e.g. Lemma~3.1 in~\cite{DaRi17a}). 
In particular, \emph{the unstable manifolds $(W^u(\Lambda_j))_{j=1,\ldots, K}$ form a partition of $M$,} i.e.
$$M=\bigcup_{j=1}^KW^u(\Lambda_j),\ \ \text{and}\ \ \forall i\neq j,\ W^u(\Lambda_i)\cap W^u(\Lambda_j)=\emptyset.$$
The same of course holds for stable manifolds. One of the main feature of such flows is the following result which is due to Smale~\cite{Sm60, Sm67}:
\begin{theo}[Smale]\label{t:smale} Suppose that $\varphi^t$ is a Morse-Smale flow. Then, for every $1\leq j\leq K$, the closure of $W^u(\Lambda_j)$ is the union of certain $W^u(\Lambda_{j'})$. 
Moreover, if we say that $W^u(\Lambda_{j'})\leqq 
 W^u(\Lambda_j)$ if $W^u(\Lambda_{j'})$ is contained in the closure of $W^u(\Lambda_{j})$, then, $\leqq$ is a partial ordering.
Finally if $W^u(\Lambda_{j'})\leqq 
W^u(\Lambda_j)$, then $\operatorname{dim}W^u(\Lambda_{j'})\leq \operatorname{dim}W^u(\Lambda_{j}).$
\end{theo}
The partial order relation on the collection of subsets $W^u(\Lambda_j)_{j=1}^K$ defined above is called \textbf{Smale causality relation}. Following Smale, 
we define an oriented graph\footnote{This diagram is the Hasse diagram associated to the poset $\left(W^u(\Lambda_j)_{j=1}^K,\leqq\right)$.} 
$D$ whose $K$ vertices are given by $W^u(\Lambda_j)_{j=1}^K$. Two vertices $W^u(\Lambda_j), W^u(\Lambda_i) $ are connected
by an oriented path starting at $W^u(\Lambda_j)$ and ending at $W^u(\Lambda_i)$ iff $W^u(\Lambda_j)\leqq W^u(\Lambda_i)$. Recall from the works of 
Peixoto~\cite{Pe62} that Morse-Smale flows form an open and dense of all smooth vector fields on surfaces while it is an open set in higher 
dimensions~\cite{Pa68}.

In the constructions from~\cite{DaRi17a, DaRi17b}, we needed to make extra assumptions on our flows, namely that they are linearizable 
near every critical element $\Lambda_i$. 
More precisely, 
we fix $1\leq l\leq \infty$ and we say that the Morse-Smale flow is \textbf{$\ml{C}^l$-linearizable} if for every $1\leq i\leq k$, the following hold:
\begin{itemize}
 \item If $\Lambda_i$ is a fixed point, there exists a $\ml{C}^l$ diffeomorphism 
$h: B_n(0,r)\rightarrow W$ (where $W$ is a small open neighborhood of $\Lambda_i$ and $B_n(0,r)$ is a small ball of radius $r$ 
centered at $0$ in $\IR^n$) and a linear map $A_i$ on $\IR^n$ such that $V\circ h=dh\circ L$ where $V$ is the vector field generating $\varphi^t$ and where
$$L(x)=A_ix.\partial_x.$$
 \item If $\Lambda_i$ is closed orbit of period $\ml{P}_{\Lambda_i}$ if there exists a $\ml{C}^l$ diffeomorphism 
$h: B_{n-1}(0,r)\times\IR/(\ml{P}_{\Lambda_i}\IZ)\rightarrow W$ (where $W$ is a small open neighborhood of $\Lambda_i$ and $r>0$ is small) and a 
smooth map $A:\IR/(\ml{P}_{\Lambda_i}\IZ)\rightarrow M_{n-1}(\IR)$ such that $V\circ h=dh\circ L$ with
$$L_i(x,\theta)=A_i(\theta)x.\partial_x+\partial_{\theta}.$$
\end{itemize}
In other words, the flow can be put into a normal form in a certain chart of class $\ml{C}^l$. We shall say that \textbf{a Morse-Smale flow is 
$\ml{C}^l$-diagonalizable} 
if it is $\ml{C}^k$-linearizable and if, for every critical element $\Lambda$, either the linearized matrix $A\in GL_n(\IR)$ or the monodromy matrix 
$M$ associated with $A(\theta)$ is diagonalizable in $\IC$. Such properties are satisfied as soon as certain (generic) non 
resonance assumptions are made on the Lyapunov exponents thanks to the Sternberg-Chen Theorem~\cite{Ne69, WWL08}. We refer to the 
appendix of~\cite{DaRi17a} for a detailed description of these nonresonant assumptions.

\subsection{Anosov flows}

We say that a flow $\varphi^t:M\rightarrow M$ $\Lambda$ is of \textbf{Anosov type}~\cite{Ano67, AbMa08} if there exist $C>0$ and $\chi>0$ and a family of spaces
$E_u(\rho),E_s(\rho)\subset T_{\rho}\mathbb{M}$ (for every $\rho$ in $M$) satisfying the following properties, 
for every $\rho$ in $M$ and for every $t\geq 0$,
\begin{enumerate}
 \item $T_{\rho}\mathbb{M}=\IR V(\rho)\oplus E_u(\rho)\oplus E_s(\rho)$ with $V(\rho)=\frac{d}{d}(\varphi^t(\rho))\rceil_{t=0}$,
 \item $d_{\rho}\varphi^t E_{u/s}(\rho)=E_{u/s}(\varphi^t(\rho))$,
 \item for every $v$ in $E_u(\rho)$, $\|d_{\rho}\varphi^{-t}v\|\leq Ce^{-\chi t}\|v\|$,
 \item for every $v$ in $E_s(\rho)$, $\|d_{\rho}\varphi^{t}v\|\leq Ce^{-\chi t}\|v \|$.
\end{enumerate}
Again, it is known from the works of Anosov that such flows form an open set inside smooth vector fields~\cite{Ano67}.


\begin{thebibliography}{99}
\bibitem{AbMa08} R.~Abraham, J.E.~Marsden \emph{Foundations of Mechanics, Second Edition} AMS Chelsea Publishing, AMS, Providence, Rhode Island (2008)
\bibitem{Ano67} D.V.~Anosov, \emph{Geodesic flows on closed Riemannian manifolds of negative curvature}, Trudy Mat. Inst. Steklov. $\mathbf{90}$ (1967)
\bibitem{Ap76} T.M.~Apostol \emph{Introduction to analytic number theory}, Undergraduate Texts in Mathematics, Springer-Verlag, New York-Heidelberg (1976) 
\bibitem{Ba16} V.~Baladi \emph{Dynamical Zeta Functions and Dynamical Determinants for Hyperbolic Maps -- A Functional Approach}, avalaible at 
https://webusers.imj-prg.fr/$\sim$viviane.baladi/baladi-zeta2016.pdf (2016)
\bibitem{BDIP96}  J.~Bertin, J.P.~Demailly, L.~Illusie and C.~Peters \emph{Introduction to Hodge Theory}, SMF/AMS Texts and Monographs, Vol.~8 (1996)
\bibitem{BisZhang92} J.M.~Bismut, W.~Zhang, \emph{An Extension of a Theorem of Cheeger and M\"uller},
Ast\'erisque 205, Soci\'et\'e Math. de France, Paris (1992)
\bibitem{burg96} D.~Burghelea, L.~Friedlander and T.~Kappeler, \emph{Asymptotic expansion of the Witten deformation of the analytic torsion}, 
Journal of Functional Analysis, 137 (2), 320--363
\bibitem{BuLi07} O.~Butterley, C.~Liverani, \emph{Smooth Anosov flows: correlation spectra and stability}, J. Mod. Dyn. 1, 301--322 (2007) 
\bibitem{Co78} C.~Conley, \emph{Isolated invariant sets and the Morse index}, CBMS Regional Conference Series in 
Mathematics 38, AMS, Providence, R.I., (1978)
\bibitem{DaRi16} N.V.~Dang, G.~Rivi\`ere \emph{Spectral analysis of Morse-Smale gradient flows}, Preprint arXiv:1605.05516 (2016)
\bibitem{DaRi17a} N.V.~Dang, G.~Rivi\`ere \emph{Spectral analysis of Morse-Smale flows I: construction of the anisotropic spaces}, Preprint (2017)
\bibitem{DaRi17b} N.V.~Dang, G.~Rivi\`ere \emph{Spectral analysis of Morse-Smale flows II: resonances and resonant states}, Preprint (2017)
\bibitem{DyGu14} S.~Dyatlov, C.~Guillarmou, \emph{Pollicott-Ruelle resonances for open systems}, Ann. H. Poincar\'e 17 (2016), 3089--3146 
\bibitem{DyZw13} S.~Dyatlov, M.~Zworski, \emph{Dynamical zeta functions for Anosov flows via microlocal analysis},  Ann. Sci. ENS 49 (2016), 543--577
\bibitem{DyZw16} S.~Dyatlov, M.~Zworski, \emph{Ruelle zeta function at zero for surfaces}, Preprint arXiv:1606.04560 (2016)
\bibitem{EnNa00} K.J.~Engel, R.~Nagel, \emph{One-Parameter Semigroups for Linear Evolution Equations}, Grad. Texts in Math. 194, Springer-Verlag New York (2000)
\bibitem{FaRoSj08} F.~Faure, N.~Roy, J.~Sj\"ostrand, \emph{Semi-classical approach for Anosov diffeomorphisms and Ruelle resonances}, Open Math. Journal, vol. 1, 35--81, (2008)
\bibitem{FaSj11} F.~Faure, J.~Sj\"ostrand, \emph{Upper bound on the density of Ruelle resonances for Anosov flows}, Comm. in Math. Physics, vol. 308, 325--364, (2011)
\bibitem{FaTs13} F.~Faure, M.~Tsujii, \emph{Band structure of the Ruelle spectrum of contact Anosov flows}, CRAS Vol. 351,  385--391 (2013)
\bibitem{FaTs17}  F.~Faure, M.~Tsujii, \emph{The semiclassical zeta function for geodesic flows on negatively curved manifolds}, 
Arxiv arXiv:1311.4932, Inv. Math. to appear (2017)
\bibitem{Fr82} J.M.~Franks \emph{Homology and dynamical systems}, CBMS Regional Conference Series in Mathematics 49, 
Conference Board of the Mathematical Sciences, Washington, D.C., AMS, Providence, R. I., (1982)
\bibitem{F87} D.~Fried \emph{Lefschetz formulas for flows}, Contemporary Mathematics Vol. 58, Part III (1987), 19--69 
\bibitem{Go15} S.~Gou\"ezel, \emph{Spectre du flot g\'eod\'esique en courbure n\'egative [d'apr\`es F. Faure et M. Tsujii]},  S\'eminaire Bourbaki (2015)
\bibitem{GuSt90} V.~Guillemin, S.~Sternberg \emph{Geometric asymptotics}, 2nd edition, Math. Surveys and Monographs 14, AMS (1990)
\bibitem{HaLa00} F.R.~Harvey, H.B.~Lawson, \emph{Morse theory and Stokes Theorem}, Surveys in Diff. Geom. VII (2000), 259--311
\bibitem{HaLa01} F.R.~Harvey, H.B.~Lawson, \emph{Finite volume flows and Morse theory}, Ann. of Math. Vol. 153 (2001), 1--25
\bibitem{HeSj85} B. Helffer, J.~Sj\"ostrand, \emph{Points multiples en m\'ecanique semi-classique IV, \'etude du complexe de Witten}, CPDE 10 (1985), 245--340
\bibitem{HeSj86} B. Helffer, J.~Sj\"ostrand, \emph{R\'esonances en limite semi-classique}, M\'em. Soc. Math. France 24-25 (1986), 1--228
\bibitem{Lau92} F. Laudenbach, \emph{On the Thom-Smale complex}, in \emph{An Extension of a Theorem of Cheeger and M\"uller}, by J.-M. Bismut and W. Zhang, 
Ast\'erisque 205, Soci\'et\'e Math. de France, Paris (1992)
\bibitem{Lee09} J.M.~Lee \emph{Manifolds and differential geometry}, Grad. Studies in Math. 107, AMS, Providence Rhode Island (2009) 
\bibitem{Mn14} P.~Mnev, \emph{Lecture notes on torsion}, arXiv:1406.3705v1 (2014)
\bibitem{Mo25} M.~Morse \emph{Relations between the critical points of a real function of n independent variables}, Trans. AMS 27 (1925), 345--396.
\bibitem{Mu93} W.~M\"uller \emph{Analytic torsion and $R$-torsion for unimodular representations}, J. Amer. Math. Soc. 6 (1993), 721--753
\bibitem{Ne69} E.~Nelson, \emph{Topics in dynamics. I: Flows.}, Mathematical Notes. Princeton University Press, Princeton, N.J.; University of Tokyo Press, Tokyo (1969) iii+118 pp. 
\bibitem{Ni08} L.I.~Nicolaescu \emph{The Reidemeister Torsion of 3-Manifolds}, De Gruyter Studies in Mathematics 30 (2003)
\bibitem{Pa68} J.~Palis \emph{On Morse-Smale dynamical systems}, Topology $\mathbf{8}$ (1968), 385--404
\bibitem{PaDeMe82} J.~Palis, W.~de Melo \emph{Geometric theory of dynamical systems. An introduction.} Springer--Verlag, New York-Berlin (1982)
\bibitem{Pe62} M.M.~Peixoto \emph{Structural stability on two-dimensional manifolds},  Topology Vol. 1 (1962), 101--120
\bibitem{Po85} M.~Pollicott, \emph{On the rate of mixing of Axiom A flows}, Invent. Math. 81 (1985), no. 3, 413--426
\bibitem{RaySi71} D.B.~Ray, I.M.~Singer \emph{$R$-torsion and the Laplacian on Riemannian manifolds}, Adv. in Math. 7 (1971), 145--210
\bibitem{dRh80} G.~de Rham, \emph{Differentiable manifolds: Forms, Currents, Harmonic Forms}, Springer (1980)
\bibitem{Ru76} D.~Ruelle, \emph{Generalized zeta-functions for Axiom A basic sets}, Bull. AMS 82 (1976), 153--156
\bibitem{Ru87a} D.~Ruelle, \emph{Resonances for Axiom A flows}, J. Differential Geom. 25 (1987), no. 1, 99--116
\bibitem{Sm60} S.~Smale, \emph{Morse inequalities for a dynamical system}, Bull. Amer. Math. Soc., 66 (1960), 43--49.
\bibitem{Sm67} S.~Smale \emph{Differentiable dynamical systems}, Bull. AMS $\mathbf{73}$, 747--817 (1967)
\bibitem{Th49} R.~Thom, \emph{Sur une partition en cellules associ\'ee \`a une fonction sur une vari\'et\'e}, CRAS 228 (1949), 973--975
\bibitem{Ts12} M.~Tsujii, \emph{Contact Anosov flows and the Fourier-Bros-Iagolnitzer transform}, Erg. Th. Dyn. Sys. 32 (2012), no. 6, 2083--2118
\bibitem{WWL08} P.~Wang, H.~Wu, W.G.~Li \emph{Normal forms for periodic orbis of real vector fields}, Acta Math. Sinica, English Series $\textbf{24}$ (2008), 797--808
\bibitem{WhWa96} E.T. Whittaker, G.N.~Watson \emph{A course of modern analysis}, Cambridge University Press, 4th Edition (1996)
\bibitem{Wi82} E.~Witten, \emph{Supersymmetry and Morse theory}, J. Diff. Geom. 17 (1982), 661--692
\end{thebibliography}
\end{document}